\documentclass[a4paper, 11pt, reqno]{amsart}
\usepackage{amssymb,slashed, color,array,tikz, verbatim, enumerate}
\usepackage{amsmath,amsfonts,amsthm,euscript,mathrsfs,multirow,bbding}
\usepackage[all]{xy}
\usepackage[english]{babel}
\usepackage{epsf,cite}
\usepackage{graphicx}
\csname @addtoreset\endcsname{equation}{section}
\newcolumntype{C}{>{\centering\arraybackslash$}p{\linewidth}<{$}}

\newcommand{\Sl}{\mathop{\rm {}SL} }

\newcommand{\nodal}{
\draw (1.9,0) .. controls  (2,.6)   and (2.6,.8) .. (3.3,-.9) ;
\draw (1.9,0) .. controls  (2,-.6)  and (2.6,-.8) .. (3.3,.9) ;
}

\newtheorem{theorem}{Theorem}[section]
\newtheorem{lemma}[theorem]{Lemma}

\theoremstyle{definition}

\theoremstyle{remark}
\newtheorem{remark}[theorem]{Remark}

\DeclareFontFamily{U}{rsf}{} \DeclareFontShape{U}{rsf}{m}{n}{ <5> <6> rsfs5 <7> <8> <9> rsfs7 <10-> rsfs10}{}
\DeclareMathAlphabet\Scr{U}{rsf}{m}{n}
\definecolor{pink}{rgb}{1,0,1}
\usepackage[OT2,T1]{fontenc}
\DeclareSymbolFont{cyrletters}{OT2}{wncyr}{m}{n}
\DeclareMathSymbol{\Sha}{\mathalpha}{cyrletters}{"58}
  \usepackage[pdftex]{hyperref}

\newcommand{\cusp}{
\draw (0,0) .. controls  (1,.2)   and (1.6,.4) .. (2,1) ;
\draw (0,0) .. controls  (1,-.2)   and (1.6,-.4) .. (2,-1) ;
}

\newcommand{\KodairaIV}{
\path (0,0)  coordinate (O);
\path (1*60:1cm) coordinate (P1);
\path (2*60:1cm) coordinate (P2);
\path (3*60:1cm) coordinate (P3);
\path (4*60:1cm) coordinate (P4);
\path (5*60:1cm) coordinate (P5);
\path (6*60:1cm) coordinate (P6);
\draw  (P1)-- (P4);\draw  (P2)-- (P5);\draw  (P3)-- (P6);
}

\newcommand{\KodairaIIICubic}{
\draw (0,0) ellipse (22pt and 17pt);
\draw (-22*.7 pt,17*.7 pt) -- +(30:30pt);
\draw (-22*.7 pt,17*.7 pt) -- +(210:30pt);
}

\newcommand{\KodTriangle}{
\draw (-1,0)--(1,0);
\draw (-.85,-.15)--(.1,1.1);
\draw (.85,-.15)--(-.1,1.1);
}

\newcommand{\KodDeux}{
\draw (0,0) ellipse (22pt and 17pt);
\draw (-22*.7 pt,17*.7 -20 pt) -- +(30:50pt);
\draw (-22*.7 pt,17*.7-20  pt) -- +(210:20pt);
}

\newcommand{\TypeTdeux}{
\draw (-1,0) -- (1,0) node[below]{\footnotesize $1$} ;
\draw (0,-1) -- (0,1) node[left]{\footnotesize 2} ;
}

\newcommand{\TypeTriple}{
 \draw (-1,-1/2) -- (1,1/2)  node [right]{ \footnotesize 3};
}

\begin{document}

\newsavebox{\cPolyTt}
\newsavebox{\cPolyTq}
\newsavebox{\cPolyTh}
\newsavebox{\cPolyTtDual}
\newsavebox{\cPolyTqDual}
\newsavebox{\cPolyQqa}
\newsavebox{\cPolyQqb}
\newsavebox{\cPolyQqaDual}
\newsavebox{\cPolyQqbDual}
\newsavebox{\cPolyQp}
\newsavebox{\cPolyQpDual}
\newsavebox{\cPolyQh}
\newsavebox{\cPolyPp}
\newsavebox{\cPolyPpDual}
\newsavebox{\cPolyPh}
\newsavebox{\cPolyH}
\newsavebox{\cQsbis}
\newsavebox{\cThbis}
\newsavebox{\cQsbbis}
\newsavebox{\cThbbis}
\newsavebox{\cQhoctabis}
\newsavebox{\cTninebis}
\newsavebox{\cQhbis}
\newsavebox{\cQhbbis}
\savebox{\cPolyPp}{
\begin{tikzpicture}[
vertex/.style={circle,draw=black!70, fill=black!70,thick, inner sep=0pt,minimum size=5pt}, 
origin/.style={circle,draw=black!70, fill=white!10,thick, inner sep=0pt,minimum size=5pt}, 
scale=.45
]
\draw[step=1cm,color=gray, very thin, dashed] (-2,-2) grid (2,2);
\node (a3) at ( 1,1) [vertex] {};
\node (a2) at ( 0,1) [vertex] {};
\node  (b1) at ( -1,0) [vertex] {};
\node (b2) at (0,0) [origin] {};
\node (b3) at ( 1,0) [vertex] {};
\node (c2) at ( 0,-1) [vertex] {};
\draw[color=black, very thick] (a2)-- (a3) --(b3)-- (c2)--(b1)--(a2);
\end{tikzpicture}
}

\savebox{\cPolyTt}{
\begin{tikzpicture}
[vertex/.style={circle,draw=black!70, fill=black!70,thick, inner sep=0pt,minimum size=5pt},origin/.style={circle,draw=black!70, fill=white!10,thick, inner sep=0pt,minimum size=5pt}, scale=.45]
\draw[step=1cm,color=gray, very thin, dashed] (-2,-2) grid (2,2);
\node (a3) at ( 1,1) [vertex] {};
\node  (b1) at ( -1,0) [vertex] {};
\node (b2) at (0,0) [origin] {};
\node (c2) at ( 0,-1) [vertex] {};
\draw[color=black, very thick] -- (a3) -- (b1)--(c2)--(a3);
\end{tikzpicture}
}

\savebox{\cPolyQqa}{
\begin{tikzpicture}
[vertex/.style={circle,draw=black!70, fill=black!70,thick, inner sep=0pt,minimum size=5pt},origin/.style={circle,draw=black!70, fill=white!10,thick, inner sep=0pt,minimum size=5pt}, scale=.45]
\draw[step=1cm,color=gray, very thin, dashed] (-2,-2) grid (2,2);
\node (a2) at ( 0,1) [vertex] {};
\node  (b1) at ( -1,0) [vertex] {};
\node (b2) at (0,0) [origin] {};
\node (b3) at ( 1,0) [vertex] {};
\node (c2) at ( 0,-1) [vertex] {};
\draw[color=black, very thick] -- (a2) -- (b1) -- (c2)--(b3)--(a2);
\end{tikzpicture}
}

\savebox{\cPolyQqb}{
\begin{tikzpicture} 
[vertex/.style={circle,draw=black!70, fill=black!70,thick, inner sep=0pt,minimum size=5pt},origin/.style={circle,draw=black!70, fill=white!10,thick, inner sep=0pt,minimum size=5pt}, scale=.45]
\draw[step=1cm,color=gray, very thin, dashed] (-2,-2) grid (2,2);
\node (a2) at ( 0,1) [vertex] {};
\node (a3) at ( 1,1) [vertex] {};
\node  (b1) at ( -1,0) [vertex] {};
\node (b2) at (0,0) [origin] {};
\node (c2) at ( 0,-1) [vertex] {};
\draw[color=black, very thick] (b1) -- (c2)--(a3)--(a2) --(b1);
\end{tikzpicture} 
}

\savebox{\cPolyTq}{\begin{tikzpicture}
[vertex/.style={circle,draw=black!70, fill=black!70,thick, inner sep=0pt,minimum size=5pt},origin/.style={circle,draw=black!70, fill=white!10,thick, inner sep=0pt,minimum size=5pt}, scale=.45]
\draw[step=1cm,color=gray, very thin, dashed] (-2,-2) grid (2,2);
\node (a1) at ( -1,1) [vertex] {};
\node (a2) at ( 0,1) [vertex] {};
\node (a3) at ( 1,1) [vertex] {};
\node (b2) at (0,0) [origin] {};
\node (c2) at ( 0,-1) [vertex] {};
\draw[color=black, very thick] -- (a1)  -- (c2)--(a3)--(a2) --(a1);
\end{tikzpicture}}

\savebox{\cPolyQp}{
\begin{tikzpicture}[vertex/.style={circle,draw=black!70, fill=black!70,thick, inner sep=0pt,minimum size=5pt},origin/.style={circle,draw=black!70, fill=white!10,thick, inner sep=0pt,minimum size=5pt}, scale=.45]
\draw[step=1cm,color=gray, very thin, dashed] (-2,-2) grid (2,2);
\node (a1) at ( -1,1) [vertex] {};
\node (a2) at ( 0,1) [vertex] {};
\node (a3) at ( 1,1) [vertex] {};
\node  (b1) at ( -1,0) [vertex] {};
\node (b2) at (0,0) [origin] {};
\node (c2) at ( 0,-1) [vertex] {};
\draw[color=black, very thick] -- (a1) -- (b1) -- (c2)--(a3)--(a2) --(a1);
\end{tikzpicture}
}

\savebox{\cPolyH}{
\begin{tikzpicture}[vertex/.style={circle,draw=black!70, fill=black!70,thick, inner sep=0pt,minimum size=5pt},origin/.style={circle,draw=black!70, fill=white!10,thick, inner sep=0pt,minimum size=5pt}, scale=.45]
\draw[step=1cm,color=gray, very thin, dashed] (-2,-2) grid (2,2);
\node (a2) at ( 0,1) [vertex] {};
\node (a3) at ( 1,1) [vertex] {};
\node  (b1) at ( -1,0) [vertex] {};
\node (b2) at (0,0) [origin] {};
\node (b3) at ( 1,0) [vertex] {};
\node (c1) at (-1,-1) [vertex] {};
\node (c2) at ( 0,-1) [vertex] {};
\draw[color=black, very thick] (a2)-- (a3) -- (b3) -- (c2)--(c1)--(b1)--(a2);
\end{tikzpicture}
}

\savebox{\cPolyQqaDual}{\begin{tikzpicture}
[vertex/.style={circle,draw=black!70, fill=black!70,thick, inner sep=0pt,minimum size=5pt},origin/.style={circle,draw=black!70, fill=white!10,thick, inner sep=0pt,minimum size=5pt}, scale=.45]
\draw[step=1cm,color=gray, very thin, dashed] (-2,-2) grid (2,2);
\node (a1) at ( -1,1) [vertex] {};
\node (a2) at ( 0,1) [vertex] {};
\node (a3) at ( 1,1) [vertex] {};
\node  (b1) at ( -1,0) [vertex] {};
\node (b2) at (0,0) [origin] {};
\node (b3) at ( 1,0) [vertex] {};
\node (c1) at (-1,-1) [vertex] {};
\node (c2) at ( 0,-1) [vertex] {};
\node (c3) at (1,-1) [vertex] {};
\draw[color=black, very thick] -- (a1) -- (b1) --(c1)-- (c2)--(c3)--(b3)--(a3) --(a2) --(a1);
\end{tikzpicture}}

\savebox{\cQsbbis}{\begin{tikzpicture}
[vertex/.style={circle,draw=black!70, fill=black!70,thick, inner sep=0pt,minimum size=5pt},origin/.style={circle,draw=black!70, fill=white!10,thick, inner sep=0pt,minimum size=5pt}, scale=.45]
\draw[step=1cm,color=gray, very thin, dashed] (-2,-2) grid (2,2);
\node (a1) at ( -1,1) [vertex] {};
\node (a2) at ( 0,1) [vertex] {};
\node  (b1) at ( -1,0) [vertex] {};
\node (b2) at (0,0) [origin] {};
\node (b3) at ( 1,0) [vertex] {};
\node (c1) at (-1,-1) [vertex] {};
\node (c2) at ( 0,-1) [vertex] {};
\node (d1) at (-1,-2) [vertex] {};
\draw[color=black, very thick]  --(a1) -- (b1) --(c1)--(d1)-- (c2)--(b3) --(a2)--(a1);
\end{tikzpicture}}

\savebox{\cThbbis}{\begin{tikzpicture}
[vertex/.style={circle,draw=black!70, fill=black!70,thick, inner sep=0pt,minimum size=5pt},origin/.style={circle,draw=black!70, fill=white!10,thick, inner sep=0pt,minimum size=5pt}, scale=.45]
\draw[step=1cm,color=gray, very thin, dashed] (-2,-2) grid (2,2);
\node (a1) at ( -1,1) [vertex] {};
\node  (b1) at ( -1,0) [vertex] {};
\node (b2) at (0,0) [origin] {};
\node (b3) at ( 1,0) [vertex] {};
\node (c1) at (-1,-1) [vertex] {};
\node (c2) at ( 0,-1) [vertex] {};
\node (d1) at (-1,-2) [vertex] {};
\draw[color=black, very thick]  --(a1) -- (b1) --(c1)--(d1)-- (c2)--(b3) --(a1);
\end{tikzpicture}}

\savebox{\cPolyPh}{\begin{tikzpicture}
[vertex/.style={circle,draw=black!70, fill=black!70,thick, inner sep=0pt,minimum size=5pt},origin/.style={circle,draw=black!70, fill=white!10,thick, inner sep=0pt,minimum size=5pt}, scale=.45]
\draw[step=1cm,color=gray, very thin, dashed] (-2,-2) grid (2,2);
\node (p0) at (0,0) [origin] {};
\node (p1) at ( 0,1) [vertex] {};
\node (p2) at ( -1,1) [vertex] {};
\node  (p3) at ( -1,0) [vertex] {};
\node (p6) at ( 0,-1) [vertex] {};
\node (p7) at ( 1,0) [vertex] {};
\node (p9) at ( 1,1) [vertex] {};
\draw[color=black, very thick] -- (p1) -- (p2) --(p3)--(p6) --(p7)--(p9)--(p1);
\end{tikzpicture}}
\savebox{\cPolyQh}{\begin{tikzpicture}
[vertex/.style={circle,draw=black!70, fill=black!70,thick, inner sep=0pt,minimum size=5pt},origin/.style={circle,draw=black!70, fill=white!10,thick, inner sep=0pt,minimum size=5pt}, scale=.45]
\draw[step=1cm,color=gray, very thin, dashed] (-2,-2) grid (2,2);
\node (p0) at (0,0) [origin] {};
\node (p1) at ( 0,1) [vertex] {};
\node (p2) at ( -1,1) [vertex] {};
\node  (p3) at ( -1,0) [vertex] {};
\node (p4) at (-1,-1) [vertex] {};
\node (p6) at ( 0,-1) [vertex] {};
\node (p9) at ( 1,1) [vertex] {};
\draw[color=black, very thick] -- (p1) -- (p2) --(p3)-- (p4)--(p6)--(p9)--(p1);
\end{tikzpicture}}
\savebox{\cPolyPpDual}{\begin{tikzpicture}
[vertex/.style={circle,draw=black!70, fill=black!70,thick, inner sep=0pt,minimum size=5pt},origin/.style={circle,draw=black!70, fill=white!10,thick, inner sep=0pt,minimum size=5pt}, scale=.45]
\draw[step=1cm,color=gray, very thin, dashed] (-2,-2) grid (2,2);
\node (p0) at (0,0) [origin] {};
\node (p1) at ( 0,1) [vertex] {};
\node (p2) at ( -1,1) [vertex] {};
\node  (p3) at ( -1,0) [vertex] {};
\node (p4) at (-1,-1) [vertex] {};
\node (p6) at ( 0,-1) [vertex] {};
\node (p7) at ( 1,0) [vertex] {};
\node (p9) at ( 1,1) [vertex] {};
\draw[color=black, very thick] -- (p1) -- (p2) --(p3)-- (p4)--(p6) --(p7)--(p9)--(p1);
\end{tikzpicture}}
\savebox{\cPolyTh}{\begin{tikzpicture}
[vertex/.style={circle,draw=black!70, fill=black!70,thick, inner sep=0pt,minimum size=5pt},origin/.style={circle,draw=black!70, fill=white!10,thick, inner sep=0pt,minimum size=4pt}, scale=.52]
\draw[step=1cm,color=gray, very thin, dashed] (-2,-2) grid (2,2);
\node (p0) at (-.05,0) [origin] {};
\node (p1) at ( 0,1) [vertex] {};
\node (p2) at ( -1,1) [vertex] {};
\node  (p3) at ( -1,0) [vertex] {};
\node (p4) at (-1,-1) [vertex] {};
\node (p5) at ( -1,-2) [vertex] {};
\node (p9) at ( 1,1) [vertex] {};
\draw[color=black, very thick] -- (p1) -- (p2) --(p3)-- (p4)--(p5)--(p9)--(p1);
\end{tikzpicture}}
\savebox{\cPolyQpDual}{\begin{tikzpicture}
[vertex/.style={circle,draw=black!70, fill=black!70,thick, inner sep=0pt,minimum size=5pt},origin/.style={circle,draw=black!70, fill=white!10,thick, inner sep=0pt,minimum size=5pt}, scale=.4]
\draw[step=1cm,color=gray, very thin, dashed] (-2,-2) grid (2,2);
\node (p0) at (0,0) [origin] {};
\node (p1) at ( 0,1) [vertex] {};
\node (p2) at ( -1,1) [vertex] {};
\node  (p3) at ( -1,0) [vertex] {};
\node (p4) at (-1,-1) [vertex] {};
\node (p5) at ( -1,-2) [vertex] {};
\node (p6) at ( 0,-1) [vertex] {};
\node (p9) at ( 1,1) [vertex] {};
\draw[color=black, very thick] -- (p1) -- (p2) --(p3)-- (p4)--(p5)--(p6) --(p9)--(p1);
\end{tikzpicture}}
\savebox{\cPolyQqbDual}{\begin{tikzpicture}
[vertex/.style={circle,draw=black!70, fill=black!70,thick, inner sep=0pt,minimum size=5pt},origin/.style={circle,draw=black!70, fill=white!10,thick, inner sep=0pt,minimum size=5pt}, scale=.45]
\draw[step=1cm,color=gray, very thin, dashed] (-2,-2) grid (2,2);
\node (p0) at (0,0) [origin] {};
\node (p1) at ( 0,1) [vertex] {};
\node (p2) at ( -1,1) [vertex] {};
\node  (p3) at ( -1,0) [vertex] {};
\node (p4) at (-1,-1) [vertex] {};
\node (p5) at ( -1,-2) [vertex] {};
\node (p6) at ( 0,-1) [vertex] {};
\node (p7) at ( 1,0) [vertex] {};
\node (p9) at ( 1,1) [vertex] {};
\draw[color=black, very thick] -- (p1) -- (p2) --(p3)-- (p4)--(p5)--(p6) --(p7)--(p9)--(p1);
\end{tikzpicture}}
\savebox{\cPolyTtDual}{\begin{tikzpicture}
[vertex/.style={circle,draw=black!70, fill=black!70,thick, inner sep=0pt,minimum size=5pt},origin/.style={circle,draw=black!70, fill=white!10,thick, inner sep=0pt,minimum size=5pt}, scale=.42]
\draw[step=1cm,color=gray, very thin, dashed] (-2,-2) grid (2,2);
\node (p0) at (0,0) [origin] {};
\node (p1) at ( 0,1) [vertex] {};
\node (p2) at ( -1,1) [vertex] {};
\node  (p3) at ( -1,0) [vertex] {};
\node (p4) at (-1,-1) [vertex] {};
\node (p5) at ( -1,-2) [vertex] {};
\node (p6) at ( 0,-1) [vertex] {};
\node (p7) at ( 1,0) [vertex] {};
\node (p8) at ( 2,1) [vertex] {};
\node (p9) at ( 1,1) [vertex] {};
\draw[color=black, very thick] -- (p1) -- (p2) --(p3)-- (p4)--(p5)--(p6) --(p7)--(p8)--(p9)--(p1);
\end{tikzpicture}}
\savebox{\cPolyTqDual}{\begin{tikzpicture}
[vertex/.style={circle,draw=black!70, fill=black!70,thick, inner sep=0pt,minimum size=5pt},origin/.style={circle,draw=black!70, fill=white!10,thick, inner sep=0pt,minimum size=5pt}, scale=.45]
\draw[step=1cm,color=gray, very thin, dashed] (-2,-2) grid (2,2);
\node (a1) at ( -1,2) [vertex] {};
\node (a2) at ( 0,2) [vertex] {};
\node (a3) at ( 1,2) [vertex] {};
\node  (b1) at ( -1,1) [vertex] {};
\node (b2) at (0,1) [origin] {};
\node (c1) at (-1,0) [vertex] {};
\node (c2) at ( 0,0) [vertex] {};
\node (d1) at (-1,-1) [vertex] {};
\node (e1) at (-1,-2) [vertex] {};
\draw[color=black, very thick] -- (a1) -- (b1) --(c1)-- (d1)--(e1)--(c2)--(a3) --(a2)--(a1);
\end{tikzpicture}}
\savebox{\cQsbis}{\begin{tikzpicture}
[vertex/.style={circle,draw=black!70, fill=black!70,thick, inner sep=0pt,minimum size=5pt},origin/.style={circle,draw=black!70, fill=white!10,thick, inner sep=0pt,minimum size=5pt}, scale=.45]
\draw[step=1cm,color=gray, very thin, dashed] (-2,-2) grid (2,2);
\node (a2) at ( 0,2) [vertex] {};
\node (a3) at ( 1,2) [vertex] {};
\node  (b1) at ( -1,1) [vertex] {};
\node (b2) at (0,1) [origin] {};
\node (c1) at (-1,0) [vertex] {};
\node (c2) at ( 0,0) [vertex] {};
\node (d1) at (-1,-1) [vertex] {};
\node (e1) at (-1,-2) [vertex] {};
\draw[color=black, very thick] -- (b1) --(c1)-- (d1)--(e1)--(c2)--(a3) --(a2)--(b1);
\end{tikzpicture}}
\savebox{\cQhoctabis}{\begin{tikzpicture}
[vertex/.style={circle,draw=black!70, fill=black!70,thick, inner sep=0pt,minimum size=5pt},origin/.style={circle,draw=black!70, fill=white!10,thick, inner sep=0pt,minimum size=5pt}, scale=.45]
\draw[step=1cm,color=gray, very thin, dashed] (-2,-2) grid (2,4);
\node(aa3) at (1,3) [vertex] {};
\node (a2) at ( 0,2) [vertex] {};
\node (a3) at ( 1,2) [vertex] {};
\node  (b1) at ( -1,1) [vertex] {};
\node (b2) at (0,1) [origin] {};
\node (c1) at (-1,0) [vertex] {};
\node (c2) at ( 0,0) [vertex] {};
\node (d1) at (-1,-1) [vertex] {};
\node (e1) at (-1,-2) [vertex] {};
\draw[color=black, very thick] -- (b1) --(c1)-- (d1)--(e1)--(c2)--(a3)--(aa3) --(a2)--(b1);
\end{tikzpicture}}

\savebox{\cQhbis}{\begin{tikzpicture}
[vertex/.style={circle,draw=black!70, fill=black!70,thick, inner sep=0pt,minimum size=5pt},origin/.style={circle,draw=black!70, fill=white!10,thick, inner sep=0pt,minimum size=5pt}, scale=.45]
\draw[step=1cm,color=gray, very thin, dashed] (-2,-2) grid (2,4);
\node (aaa4) at ( 2,4) [vertex] {};
\node(aa3) at (1,3) [vertex] {};
\node (a2) at ( 0,2) [vertex] {};
\node (a3) at ( 1,2) [vertex] {};
\node  (b1) at ( -1,1) [vertex] {};
\node (b2) at (0,1) [origin] {};
\node (c1) at (-1,0) [vertex] {};
\node (c2) at ( 0,0) [vertex] {};
\node (d1) at (-1,-1) [vertex] {};
\draw[color=black, very thick] --(b1)--(c1)--(d1)--(c2)--(a3)--(aaa4)--(aa3) --(a2)--(b1);
\end{tikzpicture}}

\savebox{\cTninebis}{\begin{tikzpicture}
[vertex/.style={circle,draw=black!70, fill=black!70,thick, inner sep=0pt,minimum size=5pt},origin/.style={circle,draw=black!70, fill=white!10,thick, inner sep=0pt,minimum size=5pt}, scale=.45]
\draw[step=1cm,color=gray, very thin, dashed] (-2,-2) grid (2,4);
\node (aaa4) at ( 2,4) [vertex] {};
\node(aa3) at (1,3) [vertex] {};
\node (a2) at ( 0,2) [vertex] {};
\node (a3) at ( 1,2) [vertex] {};
\node  (b1) at ( -1,1) [vertex] {};
\node (b2) at (0,1) [origin] {};
\node (c1) at (-1,0) [vertex] {};
\node (c2) at ( 0,0) [vertex] {};
\node (d1) at (-1,-1) [vertex] {};
\node (e1) at (-1,-2) [vertex] {};
\draw[color=black, very thick] --(b1)--(c1)--(d1)--(e1)--(c2)--(a3)--(aaa4)--(aa3) --(a2)--(b1);
\end{tikzpicture}}

\savebox{\cQhbbis}{\begin{tikzpicture}
[vertex/.style={circle,draw=black!70, fill=black!70,thick, inner sep=0pt,minimum size=5pt},origin/.style={circle,draw=black!70, fill=white!10,thick, inner sep=0pt,minimum size=5pt}, scale=.45]
\draw[step=1cm,color=gray, very thin, dashed] (-2,-2) grid (2,4);
\node (aaa4) at ( 2,4) [vertex] {};
\node(aa3) at (1,3) [vertex] {};
\node (a2) at ( 0,2) [vertex] {};
\node (a3) at ( 1,2) [vertex] {};
\node (b2) at (0,1) [origin] {};
\node (c1) at (-1,0) [vertex] {};
\node (c2) at ( 0,0) [vertex] {};
\node (d1) at (-1,-1) [vertex] {};
\node (e1) at (-1,-2) [vertex] {};
\draw[color=black, very thick] --(c1)--(d1)--(e1)--(c2)--(a3)--(aaa4)--(aa3) --(a2)--(c1);
\end{tikzpicture}}

\savebox{\cThbis}
{\begin{tikzpicture}
[vertex/.style={circle,draw=black!70, fill=black!70,thick, inner sep=0pt,minimum size=5pt},origin/.style={circle,draw=black!70, fill=white!10,thick, inner sep=0pt,minimum size=5pt}, scale=.45]
\draw[step=1cm,color=gray, very thin, dashed] (-2,-2) grid (2,2);
\node (a3) at ( 1,2) [vertex] {};
\node  (b1) at ( -1,1) [vertex] {};
\node (b2) at (0,1) [origin] {};
\node (c1) at (-1,0) [vertex] {};
\node (c2) at ( 0,0) [vertex] {};
\node (d1) at (-1,-1) [vertex] {};
\node (e1) at (-1,-2) [vertex] {};
\draw[color=black, very thick] -- (b1) --(c1)-- (d1)--(e1)--(c2)--(a3) --(b1);
\end{tikzpicture}}
\begin{titlepage}
\begin{center}
\baselineskip=14pt{\LARGE 
A New Model for Elliptic Fibrations\\
  with a  Rank One Mordell-Weil Group:    \\
 I. Singular Fibers and  Semi-Stable Degenerations. 
\\    
}
\vspace{2 cm}
{\large  Mboyo Esole$^{\spadesuit,\heartsuit}$,  Monica Jinwoo~Kang${}^\heartsuit$,    and Shing-Tung Yau$^{\spadesuit}$ } \\
\vspace{1 cm}

\begin{tabular}{l}
\quad ${}^\spadesuit$Department of Mathematics, \ Harvard University, Cambridge, MA 02138, U.S.A.\\
\quad ${}^\heartsuit$Department of Physics, Harvard University, Cambridge, MA 02138, U.S.A. 
\end{tabular}

\vspace{1cm}

{\bf Abstract}
\vspace{.3 cm}
\end{center}
{\small

We introduce a new model for elliptic fibrations endowed with a Mordell-Weil group of rank one. We call it a   Q$_7(\mathscr{L},\mathscr{S})$ model. 
 It naturally generalizes several previous models of elliptic fibrations popular in the  F-theory literature.  The model is also explicitly smooth, thus relevant physical quantities can be computed in terms of topological invariants in  straight manner. Since the general fiber is defined by a cubic curve, basic arithmetic operations on the curve can be done using the chord-tangent group law.
We will use this model to determine the spectrum of singular fibers of an elliptic fibration of rank one and compute a generating function for its Euler characteristic. With a view toward string theory, we determine a semi-stable degeneration which is understood as a weak coupling limit in F-theory. We show that it satisfies a non-trivial topological relation at the level of homological Chern classes. This identity ensures that the D3 charge in F-theory is the same as the one in the weak coupling limit.

\vfill

Email:\  {\tt    esole at  math.harvard.edu,
jkang at physics.harvard.edu,  yau  at math.harvard.edu}
}

\end{titlepage}
\addtocounter{page}{1}
 \tableofcontents{}
\newpage

\section{Introduction and summary}

In the third section of {\em Enumeratio Linearum Tertii Ordinis} (The Enumeration of Cubics), Sir Isaac Newton finds that all cubic curves can be put in one of the following four canonical forms \cite{Ball,Talbot,Guicciadini,Lexicom}:
\begin{align}
    y x^2 + A x &= B y^3 + C y^2 + Dy +E,\\
     xy &=A x^3 + B x^2 + C x + D,\\
    y^2 &=A x^3 + B x^2 + C x + D,\\
         y &=A x^3+ B x^2 + C x + D.
\end{align}
 The second and fourth forms are  curves of genus zero  as  the variable $y$ is a rational function of $x$. The first and third forms are curves of genus one. 
Since they have rational points, they are actually elliptic curves. 
 The third canonical form was called a {\em cubic hyperbola} by Newton. Today, it is universally known as a Weierstrass equation. Its Newton's polygon is a reflexive triangle with six lattice points on its boundary\footnote{ A reflexive polygon  is a polygon with a unique lattice point in its interior. For a binary planar algebraic curve, the number of interior lattice points in their  Newton's polygon gives the arithmetic genus of the curve.}. Its Mordell-Weil group is generically trivial. On the other hand, the first canonical form has two rational points along the line of infinity. This indicates that its Mordell-Weil group is non-trivial. 
  The first canonical form is not as famous as the Weierstrass model. But as we shall see in this paper, it corresponds to the general form of a cubic curve with a non-trivial Mordell-Weil group. This remark is particularly powerful when this curve is used as a model for a  fibration.  The Newton's polygon of this curve is a reflexive quadrilateral with seven lattice points on its boundary. If we interpret the coefficients of this equation as parameters defined over a base space $B$, the first canonical form  describes an elliptic fibration with a Mordell-Weil group of rank one. 
  This  fibration is generically smooth and can be used as a Jacobian for any elliptic fibration with a non-trivial Mordell-Weil group.

The aim of this paper is  to introduce a new model for elliptic fibrations  endowed with a Mordell-Weil group of rank one over a variety $B$. We call this model Q$_7(\mathscr{L},\mathscr{S})$. It is given by a smooth hypersurface in a projective bundle  characterized by two line bundles $\mathscr{L}$ and $\mathscr{S}$. Its general fiber is modeled by a plane cubic whose  Newton's polygon is a reflexive quadrilateral with seven lattice points on its boundary.  It generalizes both  the E$_6$ elliptic fibration and the elliptic introduced recently by  Cacciatori, Cattaneo, and Van Geemen  \cite{CCVG}. Ultimately, it can be traced back to Newton's first form of cubic curves.

The Q$_7(\mathscr{L},\mathscr{S})$ model allows for a particularly friendly derivation of several geometric, topological and arithmetic  properties. We provide a classification of  its singular fibers. 
Following \cite{AE1,AE2,EFY}, we derive a  {\em generalized Sethi-Vafa-Witten formula} for these elliptic fibrations. This is   a generating function for the Euler characteristic over a base of arbitrary dimension. We also explicitly construct a semi-stable degeneration of the  Q$_7(\mathscr{L},\mathscr{S})$ model. This degeneration  will be understood {in F-theory} as a {\em weak coupling limit} mapping the elliptic fibration to an  orientifold theory \cite{Sen.Orientifold,Denef.LH,CDE,Esole:2012tf}. Inspired by string dualities, we prove a topological relation  in the Chow ring of such  elliptic fibrations, connecting their homological total Chern class with that of several sub-varieties naturally produced by the semi-stable degeneration. In the context of F-theory, we show that this relation induces the non-trivial fact that the D3 charge is the same in F-theory and its orientifold weak coupling limit. 
 For previous works on elliptic fibrations with non-trivial Mordell-Weil groups in F-theory, see  
 \cite{ AE2,EFY,
 Park:2011wv,
 Park:2011ji,
Cvetic:2013qsa,
 Braun:2014oya,
 Morrison:2014era,
 Mayrhofer:2014opa,
 Kuntzler:2014ila,
 Klevers:2014bqa,
 Braun:2014nva,
 Cvetic:2013jta,
 Borchmann:2013hta,
 Cvetic:2013uta,
 Braun:2013nqa,
 Grimm:2013oga,
 Cvetic:2013nia,
 Borchmann:2013jwa,
 Braun:2013yti,
 Mayrhofer:2012zy,
 Braun:2011zm
 }.

\subsection{Structure of the paper} 
In the rest of this introduction, we introduce some basic notions and notation before to give the formal definition of a Q$_7(\mathscr{L},\mathscr{S})$ elliptic fibration. We then summarize the results of the paper.
In section \ref{section.Properties}, we collect some basic properties of the Q$_7(\mathscr{L},\mathscr{S})$ elliptic fibration. In particular, we study some limits that recovered well known elliptic fibrations such as the $E_6$ and the $E_7$ models. 
We also classify the singular fibers of a non-singular fibration of type Q$_7(\mathscr{L},\mathscr{S})$. Then, We prove the theorem on its Euler characteristic.  
In section \ref{section.review}, we quickly review the notion of weak coupling limits of an elliptic fibration and its connection with F-theory on elliptic fourfolds. 
In section \ref{section.weak}, we consider the Q$_7(\mathscr{L},\mathscr{S})$ model in the context of F-theory on elliptic fourfolds and prove the existence of a weak coupling limit. We establish that this weak coupling gives an orientifold theory for which the tadpole matching condition is satisfied. 
Then, We compute the spectrum of branes and see that it corresponds to an orientifold with a  $Sp(1)$ stack and a Whitney brane.
We present our conclusions in section  \ref{section.conclusions}.

\subsection{ Basic notions and convention}
We summarize our convention and recall some basic definitions. 
Throughout this paper,   we  work over an algebraically closed field $k$ of characteristic zero. For all purposes, the reader can assume that $k$ is the field of complex numbers $\mathbb{C}$. 
\subsubsection{ Projective bundles} Given a vector bundle $V\rightarrow B$ defined over a variety $B$, we denote by $\pi:\mathbb{P}[V]\rightarrow B$ the projectivization of a vector bundle $V\rightarrow B$. This is extended to  weighted projective bundles, for which we write $\mathbb{P}_{\vec{w}}[V]$ where $\vec{w}$ are the weights of characterizing the weighted projective bundle. 
There are two conventions for projective bundles. We use the classical (pre-Grothendieck) convention in which a projective space $\mathbb{P}[V]$ is the space of directions of a vector space $V$ in contrast to the space of hyperplanes. 
This is particularly important when dealing with Chern classes. Our convention on projective bundle matches Fulton's book on intersection theory. 
 We denote by $\mathbb{P}[V]_{\mathscr{W}}$ the collection of hypersurfaces defined  as the zero locus of a section of the line bundle $\mathscr{W}$ in the projective bundle $\mathbb{P}[V]$.  
We denote by $\mathscr{O}_{\mathbb{P}[V]}(-1)$ the tautological line bundle of $\mathbb{P}[V]$. Its dual is denoted $\mathscr{O}_{\mathbb{P}[V]}(1)$. 
We denote by  $\mathscr{O}_{\mathbb{P}[V]}(n)$ the $n$th tensor product of  $\mathscr{O}_{\mathbb{P}[V]}(1)$ (if $n>0$) and of  $\mathscr{O}_{\mathbb{P}[V]}(-1)$ if $n<0$.
When the context is clear, we abuse notations and write    $\mathscr{O}(n)$ for   $\mathscr{O}_{\mathbb{P}[V]}(n)$.

\subsubsection{ Genus-one fibrations, rational sections}
A {\em genus-one fibration} over a variety $B$ is a surjective morphism $\varphi:Y\rightarrow B$ onto $B$ such that the general fiber is a smooth projective curve of  genus one. 
A {\em rational section} of the genus-one fibration is a rational map $\sigma: B\rightarrow Y$ such that the image of $\varphi \circ \sigma$ is dense in $B$ and restrict to the identity  on the domain of definition  of $\sigma$. In particular, a rational section can be ill-defined over a divisor of $B$.  The image of the base under a rational section gives a divisor of $Y$.
\subsubsection{Elliptic fibrations, Weierstrass models and Mordell-Weil group}
When the genus-one fibration admits a rational section, we call it an {\em elliptic fibration}. An elliptic fibration is birational to a  Weierstrass model \cite{MumS,Formulaire}. 
{\em A Weierstrass model } over $B$ is  an hypersurface cut out by a  section of $\mathscr{O}(3)\otimes \pi^{\ast} \mathscr{L}^6$ in the  projective bundle $\mathbb{P}[\mathscr{L}^2\oplus\mathscr{L}^3\oplus \mathscr{O}_B]\rightarrow B$, where $\mathscr{L}$ is a line bundle defined over $B$. The canonical form of a Weierstrass model in the notation of Tate is \cite{Formulaire,Tate}:
\begin{equation}
z y^2+a_1 x y z + a_3 y z^2=x^3 + a_2 x^2 z + a_ 4 x z^2 + a_6 z^3,
\end{equation}
where a coefficient $a_i$ is a  section of $\mathscr{L}^i$ for  $i=1,2,3,4,6$.
The rational  section of the Weierstrass model is then  $O: z=x=0$, which is a point of inflection on the generic fiber. 
 The group of rational sections of the elliptic fibration $\varphi:Y\rightarrow B$ is called the Mordell-Weil group $MW(\varphi)$.  It is a finitely generated Abelian group. Its rank and torsion group are birational invariants of the elliptic fibration. 
The definition of $MW(\varphi)$ assumes the choice of an identity element.
 For a Weierstrass equation, the canonical choice is the point of inflection  $O:z=x=0$. 
\subsubsection{ Quartic models of genus one curve and elliptic fibrations}
We can think of an elliptic fibration as an elliptic curve over the function field of the base. From that point of view, the Mordell-Weil group is really just the group of rational points.
Given a divisor of degree two on a genus one curve, the Riemann-Roch theorem ensures that the curve can be embedded in the weighted projective space $\mathbb{P}_{2,1,1}$ as a quartic curve
 \begin{equation}\label{Equation.QuarticE7}
 u^2=q_0 y^4 + q_1 y^3 z + q_2 y^2 z^2 + q_3 y z^3 + q_4 z^4.
 \end{equation}
For general values of the coefficients, this is a smooth genus one fibration. 
Genus one fibration of this type has been discussed recently by Braun and Morrison \cite{Braun:2014nva}. 
If we call $\mathscr{L}^2$ the line bundle over the base such that $u$ is a section of $\mathscr{O}(2)\otimes \pi_\ast\mathscr{L}^2$, 
then equation \eqref{Equation.QuarticE7} is of type $\mathbb{P}_{2,1,1}[\mathscr{L}^2\oplus\mathscr{M}\oplus \mathscr{O}_B]_{\mathscr{O}(4)\otimes \pi^\star \mathscr{L}^4}$ for some line bundle $\mathscr{M}$. Then $y$ is a section of $\mathscr{O}(1)\otimes\pi_\ast\mathscr{M}$, $z$ is a section of $\mathscr{O}(1)$, and the coefficient $q_i$ $(i=0,1,2,3,4)$ is a  section of $\mathscr{L}^4\otimes\mathscr{M}^{-i}$. 
 
\subsubsection{Fibrations with Mordell-Weil group of rank one}
  If we have a genus one curve endowed with a  divisor of degree two that  splits into two rational points, we can assume  in equation \eqref{Equation.QuarticE7}  that $q_4$ is a perfect square and the quartic equation can be put in the following canonical form:
\begin{equation}\label{Equation.quartic}
u(u + b_2 z^2 )+y(-c_0 y^3 + c_1 z y^2 + c_2 y z^2 + c_3 z^3)=0,
\end{equation}
 which has  double points singularities at $u=y=b_2=c_3=0$. These  singularities can be resolved by blowing up the non-Cartier divisor $u=y=0$.   This is a small and thus  crepant resolution\footnote{
 Since we blow up a divisor, we cannot change the canonical class. Hence, a small resolution is always crepant.
 }. It follows that any elliptic fibration with a Mordell-Weil group of rank one has a model for which the general fiber is a quartic curve in $Bl_{[0:0:1]}\mathbb{P}^2_{2,1,1}$. 
\begin{remark}
 If $b_2$ is a unit, we get an elliptic fibration of type E$_7$. This is a model of type $\mathbb{P}_{2,1,1}[\mathscr{L}^2\oplus\mathscr{L}\oplus \mathscr{O}_B]_{\mathscr{O}(4)\otimes \pi^\star \mathscr{L}^4}$. Such an elliptic fibration has generically rank one. 
 We can normalize the equation to have $q_4=1$. We can take the identity element of the Mordell-Weil group to be the rational point $y=u-z^2=0$ and the generator of the Mordell-Weil group to be $y=u+z^2=0$. These two sections do not intersect. 
 \end{remark}

 Computing the Jacobian of genus one quartic curves  is a classical problem that can be solved by invariant theory as developed already in the 19th century by Cayley. In his famous memoir \cite{Weil.Hermite},  Weil has even traced solution to this problem  to Hermite and Euler. Nagell's algorithm, which is discussed for example in chapter 8 of Cassel's book \cite{Cassels}, provides a direct computation of the birational map to the  Jacobian for a cubic with a rational point in characteristic different from two and three. 
More advanced tools are necessary when the characteristic is two or three. This problem has been treated in all generality in \cite{ARVT}.\\

Recently, Morrison and Park have revisited the geometry of  elliptic fibration with a Mordell-Weil group of rank one in the context of F-theory on Calabi-Yau threefolds \cite{MorrisonPark}. In their treatment, they navigate between two  models: the Jacobian and its resolution.  The Jacobian (given by a Weierstrass model) is useful to compute arithmetic properties of the elliptic fibration. They also need an explicit resolution of singularities to evaluate intersection numbers necessary in the  discussion of cancellations of anomalies. More generally, in F-theory on an elliptic threefold, a smooth model is useful to compute several physical quantities that are expressed in terms of topological invariants. For example, the Euler characteristic is used in the discussion of anomaly cancellations and intersection numbers are interpreted as physical charges. In F-theory on an elliptic fourfold, the D3 tadpole depends on the Euler characteristic of the fourfold. 
For all these reasons, it will be useful to have a smooth model for an elliptic fibration with Mordell Weil group of rank one. It would specially be useful if the  the generic fiber is a  cubic curve so that several arithmetic properties can be computed using the chord-tangent law without passing to the Jacobian.

\subsection{Definition of a  Q$_7(\mathscr{L},\mathscr{S})$ model}
The model of elliptic fibrations with Mordell-Weil group of rank one that we introduced in this paper is  an  hypersurface cut by a section of the line bundle $
\mathscr{O}(3)\otimes \pi^{\ast} \mathscr{L}^2\otimes\pi^{\ast} \mathscr{S}$ in the projective bundle 
$
\pi: \mathbb{P}[\mathscr{L}\oplus\mathscr{S}\oplus \mathscr{O}_B]\rightarrow B$,
where $\mathscr{L}$ and $\mathscr{S}$ are two line bundles over the base $B$.  
The defining equation can be put in the following canonical form: 
\begin{equation}
Q_7(\mathscr{L},\mathscr{S}):  \quad y (x^2 -c_0 y^2)+ z( c_1 y^2+b_2 x z+c_2  y z +c_3 z^2 )=0. \label{def.Q7}
\end{equation}
The Newton's polygon of this cubic equation is  a reflexive  quadrilateral with seven lattice points on its boundary. For this reason, we  denote this model  Q$_7(\mathscr{L},\mathscr{S})$.  

\begin{figure}[thb]

\begin{center}
\scalebox{1.3}{
\begin{tikzpicture}[scale=.3]
\node (Qpd) at (0,0) {
\begin{tikzpicture}
[vertex/.style={circle,draw=black!70, fill=black!70,thick, inner sep=0pt,minimum size=5pt},origin/.style={circle,draw=black!70, fill=white!10,thick, inner sep=0pt,minimum size=5pt}, scale=.55]
\draw[step=1cm,color=gray, very thin, dashed] (0,0) grid (2,3);
\node (p7) at (1,1) [origin] {};
\node (p1) at ( 0,1) [vertex] {};
\node (p2) at ( 0,2) [vertex] {};
\node  (p3) at ( 0,3) [vertex] {};
\node (p5) at (2,1) [vertex] {};
\node (p6) at ( 1,0) [vertex] {};
\node (p4) at ( 1,2) [vertex] {};
\node (p0) at ( 0,0) [vertex] {};
\draw[color=black, very thick] --(p0)-- (p1) -- (p2) --(p3)-- (p4)--(p5)--(p6)--(p0);
\end{tikzpicture}

};\node[left] at (Qpd.west) {Q$_7$}; 
\end{tikzpicture}}
\end{center}
\caption{Newton polygon of a Q$_7$ reflexive polytope. A  Q$_7(\mathscr{L},\mathscr{S})$ model is described by a section of a line bundle $O(3)\otimes \pi^{\ast} \mathscr{L}^2\otimes \pi^{\ast}\mathscr{S}$ in the projective bundle  $\mathbb{P}[\mathscr{L}^2\oplus \mathscr{S}\oplus\mathscr{O}_B].$ Its equation is automatically of type Q$_7$. \label{Q7cubic.pic}}

\end{figure}
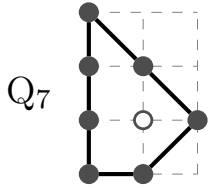

The coefficients are sections of the following line bundles: \vskip .2cm
\begin{center}
\begin{tabular}{|c|c|c|c|c|c|}
\hline
Line bundles & $\mathscr{L}^2\otimes \mathscr{S}^{-2}$ &$\mathscr{L}^2\otimes \mathscr{S}^{-1}$ &$\mathscr{L}^2$  & $\mathscr{L}^2\otimes \mathscr{S}$  & $\mathscr{L}^2\otimes \mathscr{S}^2$ \\
\hline
sections &  $c_0$ &  $ c_1$ & $ c_2$ & $c_3$ & $b_2^2$ \\
\hline
\end{tabular}
\end{center}\vspace{.5cm}

The details of the projective bundle and the section characterizing the variety are enough to ensure that we get exactly the same Jacobian as the one of an elliptic fibration of rank one. 
In the defining equation, we denote the projective coordinates as $[x:y:z]$:\\
\begin{center}
\begin{tabular}{|c|c|c|c|c|c|c|}
\hline
Line bundles & $\mathscr{O}(1)\otimes\pi^\ast \mathscr{L}$ &$\mathscr{O}(1)\otimes\pi^{\ast} \mathscr{S}$ & $\mathscr{O}(1)$ & $\mathscr{O}(3)\otimes\pi^\ast \mathscr{L}^2 \otimes\pi^\ast \mathscr{S}$\\
\hline
sections & $x$ & $y$ & $z$ & $Q_7(\mathscr{L},\mathscr{S})$\\
\hline
\end{tabular}
\end{center}\vspace{.5cm}
The special case of Q$_7(\mathscr{L},\mathscr{O}_B)$ gives the  elliptic fibration recently introduced  by  Cacciatori, Cattaneo, and Van Geemen  \cite{CCVG} while the case Q$_7(\mathscr{L},\mathscr{L})$ is the usual E$_6$ elliptic fibration \cite{AE1}.
Two rational points  of a Q$_7(\mathscr{L},\mathscr{S})$ model are 
\begin{equation} 
O:y=z=0\  \quad{and}\  \    O': y=b_2 x + c_3 z=0.
\end{equation}
These two points are  on the line $y=0$ which is tangent to the elliptic curve at $O$. 
There is also a degree-two divisor on each fiber given by  $z=x^2-c_0 y^2=0$,
 where $c_0$ is a section of $\mathscr{L}^2\otimes \mathscr{S}^{-2}$. In the case of 
Q$_7(\mathscr{L},\mathscr{L})$, $c_0$ is just a constant and we get two additional rational sections, $z=x\pm \sqrt{c_0} y=0$, which are inverse of each other in the Mordell-Weil group.

\subsection{The Jacobian fibration}
Since $z=y=0$  defines a section for the elliptic fibration $Q_7(\mathscr{L},\mathscr{S})$, we can define a birationally equivalent Weierstrass model.  
 We interpret the Jacobian of the fibration as the relative Picard scheme $\mathop{Pic}^0(Y/B)$ following \cite{ARVT} (see chapter 9 of \cite{BLR}).
Using the formula for the Jacobian of a family of plane cubics over an arbitrary base scheme \cite{ARVT}, we get the following  Weierstrass equation:
\begin{equation}
zy(y -b_2 c_1  z) = x^3 - c_2 x^2 z + (-c_0 b_2^2+ c_3 c_1) xz^2 +(c_3^2  +  b_2^2 c_2)c_0 z^3,
\end{equation}
 which admits the rational point  (see section \ref{Appendix.A}):
\begin{equation}
x=c_2+\frac{c_3^2}{b_2^2},\quad y=\frac{2c_3^3+ 2c_3 c_0 b_2^2- b_2^4 c_1 }{2b_2^3}.
\end{equation}
We can rewrite the Weierstrass form in the  short  form \cite{Formulaire, Tate}   
 $$z y^2=x^3+ F x z^2 + Gz^3,$$
 with $F$ and $G$ exactly as in \cite{MorrisonPark}  modulo the redefinition $c_2\mapsto -c_2$:
 \begin{align}\label{Jacobian}
F &=  -b_2^2 c_0 +  c_1 c_3 - \frac{1}{3} c_2^2,\quad
G = \frac{2}{3}b_2^2 c_0 c_2 + \frac{1}{4} b_2^2 c_1^2 +  c_0 c_3^2 + \frac{1}{3} c_1 c_2 c_3 - \frac{2}{27} c_2^3.
\end{align} It is interesting to see that $F$ and $G$ are respectively sections of $\mathscr{L}^4$ and $\mathscr{L}^6$.
 In particular, the line bundle $\mathscr{S}$ has disappeared.

\subsection{The spectrum of singular fibers}
The Q$_7(\mathscr{L},\mathscr{S})$ model has all the types of singular cubics with the exception of the triple line as seen on figure \ref{figure.cubic}. In particular, there is a   non-Kodaira fiber composed of two rational curves of multiplicity one and two  intersecting transversely\footnote{For non-Kodaira fibers in F-theory see \cite{EY,EFY,Esole:2014bka,Esole:2014hya,Morrison:2011mb,Hayashi:2014kca,CCVG,Cattaneo:2013vda,Braun:2013cb}.}. We call it  a fiber of type IV$^{(2)}$. Such a non-Kodaira fiber is very natural from the point of view of degeneration of genus-one curves modeled by cubic curves. It can be understood as a limiting case of a Kodaira fiber of type I$_3$  or of type IV.   See  figure \ref{figure.cubic} and table \ref{Table.E6'.fibers}.

\begin{figure}[htb]

\begin{center}
\begin{tikzpicture}
\node (a1) at (0,0){$I_0$};
\node (b1) at (2.5,-1)  {\begin{tikzpicture} [scale=.5]  \cusp  \end{tikzpicture}}       edge [<-] (a1);
\node (c1) at (2.5,1)  {\begin{tikzpicture} [scale=.5]  \nodal  \end{tikzpicture}} edge [<-] (a1);
\node  (b2) at (2*2.5,-1)  {\begin{tikzpicture} [scale=.7]  \KodairaIIICubic  \end{tikzpicture}} edge [<-] (b1)  ;
\node  (c2) at (2*2.5,1)   {\begin{tikzpicture} [scale=.7]  \KodDeux  \end{tikzpicture}}edge [<-] (c1);
\node  (b3) at (3*2.5,-1)  {\begin{tikzpicture} [scale=.7]  \KodairaIV  \end{tikzpicture}} edge [<-] (b2);
\node  (c3) at (3*2.5,1) {\begin{tikzpicture} [scale=.7]  \draw (-1,0)--(1,0);
\draw (-.85,-.15)--(.1,1.1);
\draw (.85,-.15)--(-.1,1.1);  \end{tikzpicture}}edge [<-] (c2);
\node  (a2) at (4*2.5,0)  {\begin{tikzpicture} [scale=.7]  \TypeTdeux  \end{tikzpicture}} edge [<-] (b3) edge [<-] (c3);
\node (a3) at (5*2.5,0)   {\begin{tikzpicture} [scale=.7]  \TypeTriple  \end{tikzpicture}} edge [<-] (a2);

\draw[<-](b1)--(c1);
\draw[<-](b2)--(c2);
\draw[<-](b3)--(c3);
\draw[loosely dotted](1.5*2.5,3)--(1.5*2.5,-3) node[above left] {$E_8$ };
\draw[loosely dotted](2.5*2.5,3)--(2.5*2.5,-3)  node[above left] {$E_7$ };
\draw[loosely dotted](3.5*2.5,3)--(3.5*2.5,-3) node[above left] {$E_6$ };
\draw[loosely dotted](4.5*2.5,3)--(4.5*2.5,-3) node[above left] {Q$_7(\mathscr{L},\mathscr{S})$};

\node (bb1) at (2.5,1.9)  {\footnotesize $II$};
\node (cc1) at (2.5,-1.9) {\footnotesize $ I_1$} ;
\node  (bb2) at (2*2.5,1.9) {\footnotesize $III$} ;
\node  (cc2) at (2*2.5,-1.9){\footnotesize $ I_2$} ;
\node  (bb3) at (3*2.5,1.9) {\footnotesize $IV $} ;
\node  (cc3) at (3*2.5,-1.9){\footnotesize $I_3 $} ;
\node  (aa2) at (4*2.5,1.3)  {\footnotesize $IV^{(2)}$} ;
\node (aa3) at (5*2.5,.8) {\footnotesize $IV^{(3)}$};
\end{tikzpicture}
\end{center}
\caption{Singular fibers of plane cubic curves. There are a total of 8 possible singular fibers  including the 6 Kodaira fibers  with at most 3 components ($I_1, II, I_2, III, I_3, IV$)  and the two non-Kodaira fibers $IV^{(2)}$  and $IV^{(3)}$. All the fibers at the left of a given dotted vertical line are those of a smooth elliptic fibration of the type ($E_8, E_7, E_6, Q_7)$ specified at the bottom left of the dotted line.   
\label{figure.cubic}
}

\end{figure}
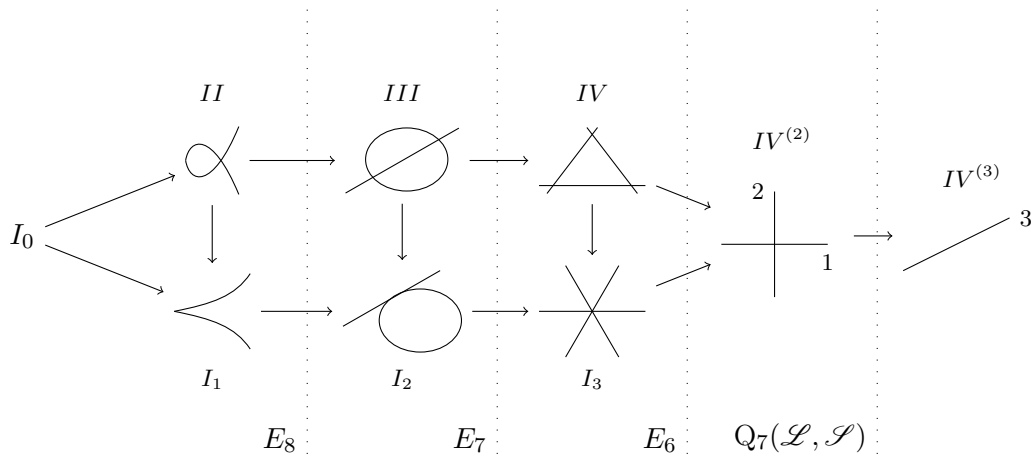

\begin{table}[htb]
\begin{center}
\begin{tabular}{|l|c|c|c|c|c|c|}
\hline
\multicolumn{1}{|c|}{Description} & Kodaira fiber & Symbols & $E_8$ & $E_7$ & $E_6$ & $Q_7$ \\
\hline
smooth genus one curve & $\checkmark$ & I$_0$ & $\checkmark$ & $\checkmark$  & $\checkmark$ & $\checkmark$  \\
\hline
 a nodal curve & $\checkmark$ &  I$_1$& $\checkmark$  & $\checkmark$ & $\checkmark$ & $\checkmark$\\
\hline
a cusp & $\checkmark$ &  II& $\checkmark$ & $\checkmark$  & $\checkmark$ & $\checkmark$\\ 
\hline
a conic and a secant  line & $\checkmark$ & I$_2$  & & $\checkmark$ & $\checkmark$ & $\checkmark$ \\
\hline  
a conic  and a   tangent line & $\checkmark$ &  III & & $\checkmark$   & $\checkmark$ & $\checkmark$ \\ 
\hline 
three lines forming a triangle & $\checkmark$ & I$_3$  &  &  & $\checkmark$ & $\checkmark$ \\ 
\hline 
three lines meeting at a point & $\checkmark$ & IV  &  &  & $\checkmark$ & $\checkmark$  \\ 
\hline  
a line and a double line &  & IV$^{(2)}$  &  &  &  & $\checkmark$ \\ 
\hline 
a triple line &  & IV$^{(3)}$  &  &  &  &  \\ 
\hline
\end{tabular} 
\end{center}
\caption{Planar  cubic curves.\label{table.cubic}  }
\end{table}

\subsection{Weak coupling limit, tadpole and flux matching condition}

From the point of view of F-theory, elliptic fibrations with a Mordell-Weil group of rank one yield an Abelian $U(1)$ gauge symmetry. 
 Depending on the dimension of the base, we are interesting at different geometric properties of the elliptic fibration. If the base is a surface, the compactification of F-theory models a six dimensional  supersymmetric gauge theory in presence of gravity. In that case, anomalies cancellations with a $U(1)$ sector are particularly subtle and have a beautiful geometric formulation.  If the base is a threefold, F-theory models a type IIB model with a non-trivial axio-dilaton profile and an Abelian $U(1)$ symmetry. Duality with M-theory gives a simple description of the origin of this $U(1)$ symmetry using the three-form of M-theory. 
 The full F-theory regime is usually strongly coupled and following Sen \cite{Sen.Orientifold}, we can explore the weak coupling limit of F-theory to make contact with well understood type IIB configurations. 
 
The weak coupling limit has a purely geometric description as a degeneration of the elliptic fibration in which the general fiber becomes  a semi-stable curve: the elliptic fibration is replaced by a ALE fibration \cite{AE2,Clingher:2012rg}. 
This description of the weak coupling limit was started  in \cite{AE2} and provides a purely geometric definition of the weak coupling limit that generalizes to any elliptic fibration regardless of the dimension of its base and independent of the Calabi-Yau condition. 
Interestingly, the topological relations that are expected to hold when F-theory is compared to its weak coupling limit are still true independently of 
all the  string theory setups necessary to make sense of them.  
For example, the D3 charge  in F-theory depends only on the Euler characteristic of the elliptic fibration and the $G_4$ flux. In a type IIB orientifold, the same D3 charge depends on the Euler characteristic of the orientifolds,  the D7 branes, the DBI fluxes supported on these D7 branes and the fluxes coming from the type IIB three form field strengths.

The weak coupling limit of a Weierstrass model is an orientifold theory \cite{Sen.Orientifold}. 
In absence of fluxes in weak and strong coupling, there is a perfect match between the D3 charge computed in type IIB and in F-theory \cite{CDE,AE1}. 
The same is true for other models of elliptic fibration such as the E$_7$, $E_6$ and D$_5$ elliptic fibrations \cite{AE2,EFY}. 
  After classifying the singular fibers of  a Q$_7(\mathscr{L}, \mathscr{S})$ model,  it is  straightforward to identify its weak coupling limit using a semi-stable degeneration following \cite{AE2}. 
The weak coupling limit is given by equation \eqref{WCL.I2.III}  which yields   the following spectrum:
$$(\star)\quad
\text{a $\mathbb{Z}_2$  orientifold}+
\text{ a Whitney brane}+
\text{an $Sp(1)$ stack}.
$$
The orientifold, the $Sp(1)$ stack, and the Whitney brane are respectively wrapping the divisors $O$, $D$, and $D_w$. 
The $Sp(1)$ stack is composed of two smooth invariant branes intersecting the orientifold transversely. 
 We prove that the tadpole matching condition is satisfied for this spectrum:
\begin{equation}
  2\chi(Y)=4\chi(O)+ 2\chi(D)+\chi^\infty(D_w).
\end{equation}
 It follows that the G-flux and the type IIB brane  fluxes match aswell \cite{Denef.LH,CDE}: 
\begin{equation}
\int_Y G_4\wedge  G_4=-\frac{1}{2} \sum_i \int_{D_i}tr(F^2).
\end{equation}
 The tadpole  matching  is a by-product of a much more general relation valid at the level of the total homological Chern classes as summarized in the following theorem: 
\begin{theorem}[Topological tadpole matching for Q$_7(\mathscr{L},\mathscr{S})$ elliptic fibrations]
A Q$_7(\mathscr{L},\mathscr{S})$ elliptic fibration endowed with the weak coupling  limit \eqref{WCL.I2.III} satisfies the topological tadpole matching condition at the level of the total Chern class: 
$$
2\varphi_* c(Y)=4\rho_* c(O)+2\rho_* c(D)+\rho_*c^\infty(D_w),
$$
where the Chern class of the Whitney brane is understood as $\rho_*c^\infty(D_w)=\rho_*c(\overline{D}_w)-\rho_*c(S)$, with $\overline{D}_w$ the normalization of $D_w$ and $S$ the cuspidial locus of the Whitney brane.
\end{theorem}
 
\section{Properties of a Q$_7(\mathscr{L},\mathscr{S})$ elliptic fibration \label{section.Properties}}

In the section, we will further explain our new model, Q$_7(\mathscr{L},\mathscr{S})$. It is similar to E$_6$ and  E$_8$ model as it is a model for a cubic.
But as we will establish in details, a  Q$_7(\mathscr{L},\mathscr{S})$ elliptic fibration has a much richer spectrum of singular fibers. It admits all possible types of singular cubics except for the triple line. In other words, its spectrum of singular fibers is  (I$_1$, II, I$_2$, III, I$_3$, IV and IV$^{(2)}$). In particular, it contains  a non-Kodaira fiber, a fiber of type $IV^{(2)}$, which consists of two rational curves of multiplicity 1 and 2 intersecting transversely at a  point. To appreciate the difference in the spectrum of singular fibers of smooth elliptic fibrations of type E$_8$, E$_7$, E$_6$, and  Q$_7(\mathscr{L},\mathscr{S})$, we review  the singular fibers of cubic plane curves and the elliptic fibrations defined in table \ref{table.cubic} and figure \ref{figure.cubic}.

The following lemma is a direct consequence of the use of the adjunction formula to compute the canonical class of an Q$_7(\mathscr{L},\mathscr{S})$ elliptic fibration: 
\begin{theorem}[Calabi-Yau condition]
An  Q$_7(\mathscr{L},\mathscr{S})$ fibration is Calabi-Yau if the line bundle $\mathscr{L}$ is the anti-canonical line bundle of the base $B$. 
\end{theorem}

\subsection{Mordell-Weil group}
An  E$_n$ elliptic fibration  ($n=8,7,6,5$ with $E_5=D_5$) has $(9-n)$ marked points defined by a  divisor  of degree $(9-n)$ on each fibers and each of the points defined by such a divisor gives a section of the elliptic fibration as the divisor splits. The elliptic fibration  Q$_7(\mathscr{L},\mathscr{S})$ has a different structure: on each fiber, we have a rational divisor of degree three as it is the case for a E$_6$ elliptic fibration. However, the divisor does not split into three rational points on every fiber but instead splits into a rational point and a divisor of degree two that does not factorize. 
This is very clear using the canonical  form \eqref{def.Q7}. The line  at infinity ($z=0$) cuts every elliptic fiber along  the following degree three divisor  
\begin{equation}z= y( x^2 - c_0 y^2)=0,
\end{equation}
which splits into a closed  point $z=y=0$ and a degree two divisor $z=x^2-c_0 y^2=0$. 
 As we circle around the  locus  $c_0$  in the base, the two points defined by  $z=x^2-c_0 y^2=0$ are exchanged. 
 This monodromy is the $\mathbb{Z}_2$  discrete group corresponding to the Galois group of the field extension needed to properly define  individually   the  two points $z=x^2-c_0 y^2=0$ over the base. 
A  Q$_7(\mathscr{L},\mathscr{S})$ elliptic fibration has an additional rational section. 
Consider the intersection with   $y=0$:
\begin{equation}
 y=0 \Longrightarrow z^2( b_2 x +c_3 z )=0 \Longrightarrow 2O+O',
\end{equation}
where 
\begin{equation}
O: y=z=0, \quad O': y=b_2 x + c_3 z=0.
\end{equation}
This indicates that $y=0$ is  tangent to the elliptic curve at $O$  and  intersects the elliptic curve at an additional point  $O'$.
These two sections intersect over the divisor $b_2=0$.

\begin{remark}
 If $\mathscr{S}=\mathscr{L}^{-1}$, $b_2$ is  a constant and thus the two sections do not intersect. 
In such a case, it is easier to  start from the following projective bundle obtained by an overall factor of $\mathscr{S}$:
\begin{equation}
\mathbb{P}[(\mathscr{L} \otimes \mathscr{S}^{-1}) \oplus \mathscr{O}_B\oplus \mathscr{S}^{-1}],
\end{equation}
and the equation is a section of $\mathscr{O}(3)\otimes \pi^\ast \mathscr{L}^2\otimes \pi^\ast \mathscr{S}^{-2}$,
which gives for $\mathscr{S}=\mathscr{L}^{-1}$: 
\begin{equation}\label{NewQ7}
\mathbb{P}[(\mathscr{L}^2 \oplus \mathscr{O}_B\oplus \mathscr{L}]_{\mathscr{O}(3)\otimes \pi^\ast \mathscr{L}^4}.
\end{equation}
\end{remark}

\begin{theorem}[Mordell-Weil group]
A  Q$_7(\mathscr{L},\mathscr{S})$ elliptic fibration has a Mordell-Weil group of rank one with  generator $O'$  and neutral element $O$. 
\end{theorem}

\subsection{A smooth hypersurface description of an elliptic fibration of rank one}
The Jacobian obtained in equation \eqref{Jacobian} is exactly the one describing a general rank one elliptic fibration as discussed in  Morrison-Park \cite{MorrisonPark} modulo the following substitution:
\begin{equation}
c_2\mapsto -c_2.
\end{equation}
This provides an interesting opportunity to obtain  a non-singular formulation of the general rank one elliptic fibration as an hypersurface in a projective bundle by generalizing the E$_6'$ fibration to have coefficients that are sections of different line bundles. 
 It is useful to notice that the  Jacobian \eqref{Jacobian} is invariant under the following scaling:
 \begin{equation}
\alpha\cdot (c_0, c_1, c_2, c_3, b_2^2)=
(\alpha^{2} c_0, \alpha c_1, c_2, \alpha^{-1} c_3 ,\alpha^{-2} b_2^2),
 \end{equation}
from which we find that the  coefficients are sections of the following line bundles\footnote{
We recall that the Weierstrass model the coefficients   
$F$ and $G$ are respectively sections of $\mathscr{L}^4$ and $\mathscr{L}^6$. 
}: 
\begin{equation}
\begin{tabular}{|c|c|c|c|c|}
\hline
$c_0$ & $c_1$ & $c_2$ & $c_3$ & $b^2_2$\\
\hline 
$\mathscr{L}^{2}\otimes\mathscr{S}^{-2}$ & $\mathscr{L}^2\otimes\mathscr{S}^{-1}$ & $\mathscr{L}^2$ & $\mathscr{L}^2\otimes\mathscr{S}$ & $\mathscr{L}^2\otimes\mathscr{S}^{2}$\\
\hline 
\end{tabular}
\end{equation}
where $\mathscr{S}$ is also a line bundle over $B$. 
We can then define a projective bundle with coordinates $[x:y:z]$ such that 
\begin{equation}
x \in\Gamma[ \mathscr{O}(1)\otimes\pi^{\ast}\mathscr{L}], \quad y \in\Gamma[\mathscr{O}(1)\otimes\pi^{\ast}\mathscr{S}], \quad z\in \Gamma[\mathscr{O}(1)],
\end{equation}
the corresponding projective bundle is $\mathbb{P}[\mathscr{L}\oplus \mathscr{S}\oplus \mathscr{O}_B]$.
With this choice, the defining equation  \eqref{def.Q7} will be a section of  the line bundle  $\mathscr{O}(3)\otimes\pi^{\ast}\mathscr{L}^2\otimes\pi^{\ast} \mathscr{S}$.
Altogether we have a new family of elliptic fibration characterized  by two line bundles $\mathscr{L}$ and $\mathscr{S}$ such that it is an hypersurface of degree   $\mathscr{O}(3)\otimes\pi^{\ast}\mathscr{L}^2\otimes\pi^{\ast} \mathscr{S}$ in the projective bundle 
$\mathbb{P}[\mathscr{L}\oplus \mathscr{S}\oplus \mathscr{O}_B]$. We call this model $Q_7(\mathscr{L},\mathscr{S})$:
\begin{equation}
Q_7(\mathscr{L},\mathscr{S}): \quad \mathbb{P}[\mathscr{L}\oplus \mathscr{S}\oplus \mathscr{O}_B]_{\mathscr{O}(3)\otimes\pi^{\ast}\mathscr{L}^2\otimes\pi^{\ast} \mathscr{S}}.
\end{equation}
All elliptic fibrations of rank one can be put in this form since it was obtained form their common Jacobian. For general values of the coefficients, this is a smooth elliptic fibration. 
A quick calculation with the adjunction formula shows that the Calabi-Yau condition for the family $Q_7(\mathscr{L},\mathscr{S})$ is $c_1(B)=c_1(\mathscr{L})$. 
\begin{theorem} An elliptic fibration of rank one is always birational to a fibration of type $Q_7(\mathscr{L},\mathscr{S})$. The fibration is Calabi-Yau when $\mathscr{L}$ is the anti-canonical line bundle of the base. 
\end{theorem}

\subsection{Special cases}
The special case of Q$_7(\mathscr{L},\mathscr{O}_B)$ gives the  elliptic fibration recently introduced  by  Cacciatori, Cattaneo, and Van Geemen  \cite{CCVG} while the case Q$_7(\mathscr{L},\mathscr{L})$ is the usual E$_6$ elliptic fibration \cite{AE1}.
The rational points of a Q$_7(\mathscr{L},\mathscr{S})$ model are 
\begin{equation} 
O:y=z=0\  \quad{and}\  \    O': y=b_2 x + c_3 z=0.
\end{equation}
These two points are  on the line $y=0$ which is tangent to the elliptic curve at $O$. 
There is also a degree-two divisor on each fiber given by  $z=x^2-c_0 y^2=0$,
 where $c_0$ is a section of $\mathscr{L}^2\otimes \mathscr{S}^{-2}$. In the case of 
Q$_7(\mathscr{L},\mathscr{L})$, $c_0$ is just a constant and we get two additional rational sections, $z=x\pm \sqrt{c_0} y=0$, which are inverse of each other in the Mordell-Weil group. 

If $\mathscr{S}=\mathscr{L}^{-1}$, we  start from the following projective bundle obtained by an overall factor of $\mathscr{S}$\footnote{Hartshorne  Chap II, Lemma 7.9.}:
\begin{equation}
\mathbb{P}[(\mathscr{L} \otimes \mathscr{S}^{-1}) \oplus \mathscr{O}_B\oplus \mathscr{S}^{-1}],
\end{equation}
and the equation is a section of $\mathscr{O}(3)\otimes \pi^\ast \mathscr{L}^2\otimes \pi^\ast \mathscr{S}^{-2}$,
which gives for $\mathscr{S}=\mathscr{L}^{-1}$: 
\begin{equation}\label{NewQ7}
\mathbb{P}[(\mathscr{L}^2 \oplus \mathscr{O}_B\oplus \mathscr{L}]_{\mathscr{O}(3)\otimes \pi^\ast \mathscr{L}^4}.
\end{equation}

\subsection{Birational map to a Jacobi quartic}
We multiply the defining equation \ref{def.Q7}, by $y$ and we replace the variable $x$ by $u=xy$. This can be explained as a birational map in the ambient projective bundle to turn it into a weighted projective bundle: 

\begin{equation}
\begin{aligned}
 \mathbb{P}[\mathscr{L}\oplus \mathscr{S}\oplus \mathscr{O}_B] &\rightarrow  \mathbb{P}_{2,1,1}[\mathscr{M}\oplus \mathscr{S}\oplus \mathscr{O}_B], \quad \mathscr{M}=\mathscr{L}\otimes \mathscr{S}\\
 [x:y:z] & \mapsto  [u:y:z]=[xy:y:z].
\end{aligned}
\end{equation}
The variable  $u$ is a section of $\pi^{\ast}\mathscr{M}\otimes \mathscr{O}(2)$. The defining  equation is a section of  $\pi^{\ast}\mathscr{M}^2\otimes \mathscr{O}(4)$: 
\begin{equation}
u(u + b_2 z^2 )+y(-c_0 y^3 + c_1 z y^2 + c_2 y z^2 + c_3 z^3)=0.
\end{equation}
This defines an elliptic fibration whose generic fiber is given by a quartic curve with  the  Newton's polygon given in figure \ref{quartic.Q7}. 
The rational sections are at $y=u=0$  and $y=u+b_2 z^2=0$.
The equation is singular at $u=y=b_2=c_3=0$. We can blow up the non-Cartier divisor $u=y=0$ to resolve with this singularity. The ambient  space becomes the blow up of the weighted projective bundle $Bl_{[0:0:1]}\mathbb{P}_{2,1,1}$.

\begin{figure}[bth]
\begin{center}
\begin{tikzpicture}
[vertex/.style={circle,draw=black!70, fill=black!70,thick, inner sep=0pt,minimum size=5pt},origin/.style={circle,draw=black!70, fill=white!10,thick, inner sep=0pt,minimum size=5pt}, scale=.8]
\draw[step=1cm,color=gray, very thin, dashed] (0,0) grid (4,2);
\node (p0) at (1,1) [origin] {};
\node (p1) at ( 0,1) [vertex] {};
\node (p2) at ( 0,2) [vertex] {};
\node  (p3) at ( 2,1) [vertex] {};
\node (p4) at (4,0) [vertex] {};
\node (p5) at ( 3,0) [vertex] {};
\node (p6) at ( 2,0) [vertex] {};
\node (p7) at ( 1,0) [vertex] {};
\draw[color=black, very thick] -- (p1) -- (p2) --(p3)-- (p4)--(p5)--(p6)--(p7)--(p1);
\end{tikzpicture}
\end{center}
\caption{Quartic Q$_7$:  a reflexive quadrilateral with seven lattice points on its boundary. 
This is the Newton's polygon  for the quartic in equation \eqref{Equation.quartic}.
 \label{quartic.Q7}}
\end{figure}
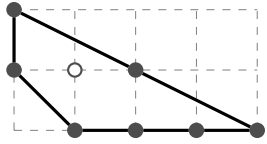

\subsection{A double cover embedded in the  Q$_7(\mathscr{L},\mathscr{S})$ fibration}
Inside a Q$_7(\mathscr{L},\mathscr{S})$ fibration, we can use the divisor of degree 2 along each fiber to define a double cover the base $X$: 
\begin{equation}X:z=x^2-c_0 y^2=0.
\end{equation}
It admits the following involution:
\begin{equation}
\sigma:\quad x\mapsto -x, 
\end{equation}
for which  $c_0=0$ is the locus of fixed point. 
In the definition of $X$, the divisor $z=0$ cuts $\mathbb{P}(E)$ along a  projective bundle $\mathbb{P}(\mathscr{L}\oplus \mathscr{S})$ with coordinates $[x:y]$. The variety $X$ naturally lives in that projective bundle with the defining equation $x^2+c_0 y^2=0$. Along  $X$, the projective coordinate  $y$ never vanishes since  otherwise we would have $z=y=x=0$, which is not possible because $[x:y:z]$ are projective coordinates of the projective bundle. It follows that we can just define $X$ in the affine patch $y\neq 0$ by taking the affine coordinate $\xi=x/ y$, which is a section of $\mathscr{L}\otimes \mathscr{S}^{-1}$:
\begin{equation}
X:   \quad \xi^2=-c_0, \quad \text{with the involution }  \xi\mapsto -\xi,
\end{equation}
which expresses  $X$ as the zero locus of a section of $\mathscr{L}^2\otimes \mathscr{S}^{-2}$.  Note that this is not an orientifold. Using  adjunction formula, it is easy to see that the  elliptic fibration  $Y\rightarrow B$ is a Calabi-Yau $(n+1)$-fold  if  and only if  $c_1(B)=c_1(\mathscr{L})$ while  the double cover $X\rightarrow B$  branched on $c_0$ is a Calabi-Yau $n$-fold if and only if 
$c_1(B)=c_1(\mathscr{L})-c_1(\mathscr{S})$.
The two are compatible only when $S$ is trivial, that is for the Q$_7(\mathscr{L},\mathscr{O}_B)$ model.
\begin{lemma}
An elliptic fibration of type Q$_7(\mathscr{L},\mathscr{O}_B)$ admits a divisor which is a double cover $X$  of the base. The Q$_7(\mathscr{L},\mathscr{O}_B)$ elliptic fibration is Calabi-Yau if and only if the double cover  $X$ of the base is also Calabi-Yau. 
\end{lemma}

\subsection{Spectrum of singular fibers of a  Q$_7(\mathscr{L},\mathscr{S})$ elliptic fibration}
 A Q$_7(\mathscr{L},\mathscr{S})$ elliptic fibration admits  up to $7$ different types of singular fibers including the non-Kodaira fiber IV$^{(2)}$ constituted of  two rational curves of multiplicity one and two intersecting transversally. These 7 different types represent all the different types of singular fibers of a cubic curve with the exception of the triple line, which is the only  multiple singular cubic curve. 

\subsubsection{General picture} 
The spectrum of singular fibers of a  Q$_7(\mathscr{L},\mathscr{S})$ elliptic fibration can be easily obtained by  a direct analysis of its defining equation. We can use the birationally equivalent Weierstrass model to characterize irreducible singular curves (nodal and cuspidial curves). Since the elliptic fibration is nonsingular, to classify  reducible singular fibers, we analyze the conditions under which the defining equation factorizes. A cubic curve can  factorize into a line and a conic. This corresponds to the fibers of type I$_2$ if the line and the conic meet at two distinct points or of type  III if the line is tangent to the conic. The conic can further splits into two lines so that the singular fiber has the structure of a triangle, Kodaira type I$_3$, or a star, Kodaira type IV. If  the conic specializes to a double line, we have a non-Kodaira fiber that we call a type IV$^{(2)}$. We summarize the condition to have all these fibers in table \ref{Table.E6'.fibers} and figure \ref{figure.cubic}. It is  also instructive to look at the spectrum of singular fibers of elliptic fibrations of type  Q$_7(\mathscr{L},\mathscr{S})$, E$_7$, E$_8$ and E$_6$ together since it reveals some beautiful patterns as illustrated in  figure \ref{figure.cubic}. 
In the next subsection, we will methodically derive table \ref{Table.E6'.fibers} .

\begin{table}[bth]
\begin{center}
\begin{tabular}{|cc|c|c|}
\hline
\multicolumn{2}{|c|}{Fiber}& $j$-invariant & Algebraic Conditions \\
\hline
\hline
I$_1$ & {\begin{tikzpicture} [scale=.3]  \nodal  \end{tikzpicture}} & $\infty$ & $\Delta=0$ \\
\hline
II & {\begin{tikzpicture} [scale=.3]  \cusp  \end{tikzpicture}}  & $0$ & $F=G=0$\\
\hline
\hline
I$_2$ & {\begin{tikzpicture} [scale=.3]  \KodDeux  \end{tikzpicture}}  & $\infty$ &   $b_2=c_3=0$ \\
\hline
III & {\begin{tikzpicture} [scale=.3]  \KodairaIIICubic  \end{tikzpicture}}   & $1728$& $c_0=c_1=c_2=c_3=0$ or $b_2=c_2=c_3=0$\\
\hline
\hline
I$_3$ 
& {\begin{tikzpicture} [scale=.3]  \KodTriangle  \end{tikzpicture}}   & $\infty$& $b_2=c_3=c_1^2+4 c_0 c_2=0$\\
\hline
IV &
 {\begin{tikzpicture} [scale=.3]  \KodairaIV \end{tikzpicture}} & $0$& $b_2=c_2=c_1=c_3=0$ \\
\hline
IV$^{(2)}$ &
 $\begin{matrix} \\ {\begin{tikzpicture} [scale=.3]  \TypeTdeux \end{tikzpicture}}\end{matrix}$  & undefined $``\frac{0}{0}"$ & $b_2=c_0=c_1=c_2=c_3=0$ \\
\hline
\hline
\end{tabular}
\end{center}
\caption{
Singular fibers of the  Q$_7(\mathscr{L},\mathscr{S})$ elliptic fibration $x^2 y +  b_2 x z^2-c_0 y^3+ c_1 z y^2+c_2 z^2 y + c_3 z^3$. 
The fiber I$_3$  and II (with $b_2=0$) are non-split when  $c_0$ is not a perfect square. The fiber III  (with $b_2\neq 0$) can also be non-split when $b_2$ is a not a perfect square.
\label{Table.E6'.fibers}
}
\end{table}

\subsubsection{Irreducible singular fibers}
Using the Weierstrass moddel, we can easily determine the irreducible singular fibers or the Q$_7(\mathscr{L},\mathscr{S})$ fibration: we have a nodal curve at a general point of the discriminant locus and a cuspidial curve at a general point of $F=G=0$:
\begin{equation}
\text{I}_1:\quad \Delta=0, \quad \text{II}:  \quad F=G=0.
\end{equation}

\subsubsection{Reducible singular fibers}
For reducible singular fibers, we  can find the condition for  factorizing the defining equation \eqref{def.Q7}. Assuming that $b_2\neq 0$, the only factorization is 
\begin{equation}
\text{III}: \quad x(x y+ b_2 z^2)=0,
\end{equation}
which requires $c_i=0$. The fiber is of type III since it is composed of two rational curves (a line and a conic) meeting at a double point. 
When $b_2=c_i=0$, the fiber III enhances to a non-Kodaira fiber  of type $IV^{(2)}$ as the degree two:
\begin{equation}
\text{IV}^{(2)}: \quad x^2 y=0.
\end{equation}
If we assume that $b_2=0$, we have much richer spectrum. In order to be able to factorize a linear term, we have to assume $b_2=c_3=0$, which gives 
\begin{equation}
\text{I}_2:\quad y(x^2 -c_0 y^2+ c_1  y z +c_2  z^2)=0.
\end{equation}
This singular fiber is constituted of two rational curves ( a line and a conic) meeting transversally at two distinct points.
  
\begin{remark}
The  section $O'$ is now the full line $y=0$. The zero section $O$ intersect only that line and does not intersect the conic. 
We also note that the two points of intersection can be seen as the intersection of the conic with the section $O'$. 
\end{remark}

We have an enhancement to a fiber of type III when the line becomes tangent to the conic. This happens when  $b_2=c_2=c_3=0$ and the line and the conic intersect at a double point:
\begin{equation}
\text{III}:\quad y(x^2 - c_0 y^2+ c_1  y z)=0.
\end{equation}
The fiber $\text{I}_2$ can enhance to a fiber $\text{I}_3$ when the conic degenerates into two  lines. This requires the additional condition $c_1^2-4 c_0 c_2=0$ so that all together we have a $\text{I}_3$ fiber when $b_2=c_3=c_1^2-4 c_0 c_2=0$ ($c_1 \neq 0$ or $c_2\neq 0$). The equation of the fiber is:
\begin{equation}
\text{I}_3:\quad y\Big(x^2-c_0(y-\frac{c_1}{2 c_0} z)^2\Big)=0  \quad or \quad y\Big(x^2+c_2(z-\frac{c_1}{2 c_2} y)^2\Big) =0.
\end{equation}
This $\text{I}_3$ fiber is split if and only if  $c_0$ is a perfect square, otherwise, we have a non-split fiber with a $\mathbb{Z}_2$ torsion. 
If $b_2=c_2=c_1=0$, we have a fiber of type IV$^{ns}$. 
\begin{equation}
\text{IV}^{ns}:\quad y(x^2-c_0 y^2)=0.
\end{equation}
This fiber is split only if $c_0$ is a perfect square. Otherwise it is a non-split fiber with a $\mathbb{Z}_2$ torsion. 
Finally, if $b_2=c_i=0$, we have a fiber of type IV$^{(2)}$.

\subsection{Generalized Sethi-Vafa-Witten formula}

The Sethi-Vafa-Witten formula gives the  Euler characteristic of a Calabi-Yau fourfolds defined by a generic Weierstrass model \cite{Sethi:1996es}. In this section,  we present a generating  function for the Euler characteristic of a  Q$_7(\mathscr{L},\mathscr{S})$ model over a base of arbitrary dimension. We also do not assume the Calabi-Yau condition. This is done using a push-forward formula following \cite{AE2,EFY}.

\begin{theorem}[Euler characteristic of Q$_7(\mathscr{L},\mathscr{S})$]\label{Theorem.Euler}
Let $L=c_1(\mathscr{L})$, $S=c_1(\mathscr{S})$ and $c_k$ be the $k$-th Chern class $c_k(TB)$ of the base $B$. 
Then, the push-forward of the total homological Chern class of the elliptic fibration $Y$ is:

\begin{align}
\begin{aligned}
\pi_\ast c(Y) &=
\frac{6 (2 L  + 2 L^2 - L S +S^2)}{(1 + 2 L - 2 S) (1 + 2 L + S)} c(B)\\
& =12 L t + (12 c_1 L - 36 L^2 + 6 L S - 6 S^2) t^2 + \\
& \quad (12 c_2 L - 
    36 c_1 L^2 + 96 L^3 + 6 c_1 L S - 36 L^2 S - 6 c_1 S^2 + 54 L S^2 - 
    6 S^3) t^3+\cdots 
    \end{aligned}
\end{align}
\end{theorem}
This gives a generating function for the Euler characteristic: if the base is of dimension $d$, then the Euler characteristic of $Y$ is given by the coefficient of $t^d$. 
If $Y= Q_7(\mathscr{L},\mathscr{S})$ is a Calabi-Yau variety, we can simplify further the expression by using $L=c_1$. 
\begin{lemma}
For  $Y$  a Calabi-Yau threefold or fourfold,  we get the Euler characteristics respectively as 
\begin{equation}
\chi(Y)=-6 (4 c_1^2 - c_1 S + S^2), \quad 
\chi(Y)=6 (10 c_1^3 + 2 c_1 c_2 - 5 c_1^2 S + 8 c_1 S^2 - S^3)
\end{equation}
\end{lemma}

\begin{proof}[Proof of  theorem \ref{Theorem.Euler}] 
The ambient space in which we define $Y$ is the  following projective bundle
\begin{equation}
\pi: \mathbb{P}(\mathscr{L}\oplus \mathscr{S}\oplus \mathscr{O}_B)\rightarrow B.
\end{equation}
We would like to compute the pushforward from its Chow's ring to the Chow's ring of the base. 
We use the fact that:
\begin{equation}
\pi_{\ast} \big(1+\zeta+\zeta^2+\cdots \big)=\frac{1}{(1+L)(1+S)}.
\end{equation}
It is easy to see that 
\begin{equation}
\frac{1}{(1+Lt )(1+St)}=\sum_{k\geq 0} (-1)^k P_k(L,S) t^k, 
\end{equation}
where we have inserting a variable $t$ to track the other. We also have:
$$P_k(L,S)=L^k+ L^{k-1} S + \cdots + S^k=\frac{L^{k+1}-S^{k+1}}{L-S}.$$
Comparing terms of the same dimension, we get: 
\begin{equation}
\pi_\ast 1=\pi_\ast \zeta =0, \quad \pi_\ast \zeta^{2+k} =(-1)^k P_k(L,S).
\end{equation}
Then we get 
\begin{equation}
\pi_\ast F(H)=\frac{F(H)-F(0)-H\partial_H F(0)}{(S-L)H}\big|_{H=-L}-
\frac{F(H)-F(0)-H\partial_H F(0)}{(S-L)H}\big|_{H=-S}.
\end{equation}
Now we apply it to the total homological Chern class of the elliptic fibration Q$_7(\mathscr{L},\mathscr{S})$: 
\begin{subequations}
\begin{equation}
c(Y)=\frac{(1+H)(1+H+L)(1+H+S)}{(1+3H+2L+S)}(3H+2L+S)c(B),
\end{equation}
which gives 
\begin{align}
\pi_\ast c(Y) &=
\frac{6 (2 L + 2 L^2 - L S +S^2)}{(1 + 2 L - 2 S) (1 + 2 L + S)} c(B)
\end{align}
\end{subequations}
\end{proof}

\section{Weak coupling limit and tadpole matching: a quick review \label{section.review}}

In a F-theory compactification on an elliptic fourfold $Y\rightarrow B$, the number of  D3 branes   ($N_{D3}$) depends only on the Euler characteristic $
\chi(Y)$ of the elliptic fibration $Y$ and the $G_4$-flux \cite{GVW}:
\begin{equation}
\text{D3 tadpole in F-theory}:\quad N_{D3}=\tfrac{1}{24}\chi(Y)-\frac{1}{2}\int_Y G_4\wedge G_4.     \label{FtheoryD3Tadpole}
\end{equation}
This  relation is derived from the duality between  M-theory and F-theory. 
For  IIB $\mathbb{Z}_2$ orientifold compactifications, the D3 tadpole depends on  fluxes and the Euler characteristics of the  cycles wrapped by the orientifolds and the D7-branes:
\begin{equation}\label{IIBD3Tadpole}
\text{D3 tadpole in type IIB}:\quad 2 N_{D3}=\tfrac{1}{6}\chi(O)+\tfrac{1}{24}\sum_i\chi(D_i)+\frac{1}{2} \sum_i \int_{D_i}tr(F^2),
\end{equation}
where $O$ is an orientifold;  $D_i$ are surfaces wrapped by D7-branes; $\int_{D_i} tr(F^2)$ are fluxes localized on the D7-branes. The trace $tr$ is taken in the adjunct representation. 
Since the number of D3 branes is independent of the string coupling and invariant under $\Sl(2,\mathbb{Z})$, one would expect a matching between  the computation of the number of D3 branes  in F-theory and in type IIB:
\begin{equation}\label{FIIB}
2\chi(Y)-24\int_Y G_4\wedge G_4=4\chi(O)+\sum_i\chi(D_i)+12 \sum_i \int_{D_i} tr(F^2).
\end{equation}
Such a matching condition was first introduced in \cite{CDE}.  For configurations such that both $G$-fluxes and type IIB fluxes  are zero, the matching of the D3 tadpole in type IIB and  in F-theory will give a purely topological relation between Euler characteristics 
\cite{CDE}: 
\begin{align}\label{FIIB.Matching}
 &\text{Tadpole  matching condition}:\quad  2\chi(Y)=4\chi(O)+\sum_i\chi(D_i).\\
\intertext{When this topological condition holds,  equation \eqref{FIIB} also gives a relation between the fluxes in F-theory and in type IIB:}
&\text{Flux matching condition} :\quad\int_Y G_4\wedge  G_4=-\frac{1}{2} \sum_i \int_{D_i}tr(F^2).
\end{align}
In general, the  curvature contribution to the D3 tadpole in F-theory and type IIB theory do not have to match. 
For example, branes seen in the type IIB limit can recombine into a different configuration of branes plus  fluxes \cite{CDE}.

\subsection{Geometric definition of a weak coupling limit}
 Following the  point of view of  \cite{AE2}, the weak coupling limit of an elliptic fibration is a degeneration such that the generic  fiber of the elliptic fibration becomes  semi-stable as we reach $\epsilon=0$.
A semi-stable elliptic curve is a singular  elliptic curve of type I$_n$. Such a singular elliptic curve has an infinite $j$-invariant. This explains the name ``weak coupling limit" as an infinite $j$-invariant means that the imaginary part of $\tau$ goes to zero which in the F-theory description of type IIB string theory essentially  means the string coupling is weak: $g_s\rightarrow 0$. 
If the semi-stable fiber is just a nodal curve I$_1$ as it is the case for a smooth Weierstrass model, we are in the case analyzed by Sen. In the case the degeneration gives a curve of type  I$_n$ with $n>1$, each irreducible components of the semi-stable curve I$_n$  describes a $\mathbb{P}^1$-bundle over the base and since two components intersects normally, all together they form a normal crossing variety $Z$. It follows that for a weak coupling limit defines with a generic fiber of type  I$_n$ naturally leads to a {\em semi-stable degeneration}. 
 We will  see it explicitly here. It is important to realize that an elliptic fibration can admit many non-equivalent weak coupling limits with different semi-stable curves I$_n$ as illustrated in \cite{AE2, EFY}.

\subsection{Brane geometry at weak coupling}
When taking a weak coupling limit, the discriminant locus can split into different components that are wrapped by orientifolds and  branes. These branes can be singular and can split further into brane-image-brane pairs in the double cover of the base. We  quickly review the most familiar ones by considering the following discriminant and $j$-invariant: 
\begin{align}
\Delta &= \epsilon^2 h^{2+n} \prod_i  (\eta^2_i-h \psi_i^2)\prod_j (\eta^2_j -h \chi_j), \prod_k \phi_k+O(\epsilon^3),\\
 j &\propto \frac{h^{4-n}}{\epsilon^2  \prod_i  (\eta^2_i-h \psi_i^2)\prod_j (\eta^2_j -h \chi_j) \prod_k \phi_k}.
\end{align}
The locus   $h=0$ is the orientifold locus as seen from the base of the elliptic fibration. 
As $\epsilon$ goes to zero, $j$ goes to infinity and the string couplings goes to zero.
$$ \lim_{\epsilon\rightarrow 0} j=\infty\Longrightarrow Im(\tau)=\infty\iff g_s=0.
$$
At weak coupling, we get an orientifold on the  double cover of the base branched at $h=0$: 
\begin{equation}
X:   \quad \xi^2=h,
\end{equation}
which defines a section of the line bundle $\mathscr{L}^2$. 
The involution $\sigma:X\rightarrow X$  which sends $\xi$ to $-\xi$ can be used to define a $\mathbb{Z}_2$ orientifold symmetry $\sigma \Omega (-)^{F_L}$ and the branched locus is therefore interpret as a $O7$ orientifold. 
The geometry of Sen's weak coupling limit can be summarized by the following diagram:
\begin{equation}
\begin{tabular}{c}\\ \\  \\ \text{Sen's limit}\\  \end{tabular}\quad \quad \xymatrix{
T^2\ar[r]& Y_{n+1}\ar[d]^{\pi} \ar[dr]^{\text{Orientifold   limit}}& \\
&   B_n & X_n \ar[l]^{1:2}_{\sigma}
}
\end{equation}
where $Y_{n+1}\rightarrow B_n$ is an elliptic fibration   and $X_n\rightarrow B$ is a double cover. 
The different terms of $\Delta$ determine different type of D7-branes. We summarize them in table \ref{table.brane} where we have also included the orientifold.

\begin{table}[htb]
\begin{center}
\begin{tabular}{|c|l|l|}
\hline
Name & In the discriminant & In the double cover $X$\\
\hline
Orientifold & $h^2$ & $\xi=0$ \\
\hline
 Whitney brane &  $\eta^2-h \chi$ & $\eta^2-\xi^2 \chi=0$ \\
\hline
 Brane-image-brane pair & $\eta^2-h \psi^2$ & $(\eta+ \xi \psi)(\eta- \xi \psi)=0$\\
\hline
 Invariant  brane &  $\eta$ & $\eta=0$\\
\hline
\end{tabular}
\end{center}
\caption{Familiar types of brane found in Sen's weak coupling limit. The Whitney brane is the one observed in Sen's limit of a $E_8$ elliptic fibration. It can specialize into a brane-image-brane pair when $\chi$ is a perfect square and into two invariant branes on top of each other when $\chi=0$.    \label{table.brane}}
\end{table}

\subsection{Generalized tadpole condition}

The D3 tadpole matching condition presented in equation \eqref{FIIB.Matching} is a topological condition that can be proven to hold in a  much more general set up than anticipated from the assumptions of its string theory origin. It generalizes to a relation valid at the level of the total homological Chern classes for  elliptic fibrations over bases of  arbitrary dimension  without even assuming the  Calabi-Yau condition \cite{AE1,AE2,EFY}:
\begin{equation}\label{GWCL}
\text{Generalized tadpole condition}:\quad 2 \varphi_\ast c(Y)=4\rho_\ast c(O)+\sum \rho_\ast c(D_i).
\end{equation}
The right-hand-side of this relation involves objects seen  in the type IIB weak coupling limit defined by taking a degeneration of the elliptic fibration while the left-hand-side is the elliptic fibration. In that respect, it requires both a  choice of an elliptic fibration and a choice of a degeneration.
  The most interesting case is the one involving a Weierstrass model (an $E_8$ elliptic fibration). In that case, the degeneration is given by the original Sen's weak coupling limit and the corresponding generalized tadpole relation was proven in  \cite{AE1}. 
Sen's weak coupling limit was generalized geometrically in \cite{AE2}. The method of \cite{AE2} provides an easy way to define a weak coupling limit for families of elliptic fibrations that are not given by a Weierstrass model. It also gives a natural way to organize such limits using the fiber geometry of the elliptic fibration. A generalized tadpole relation is available for  $E_8$, $E_6$ and $E_7$ elliptic fibration  \cite{AE1,AE2}, and $D_5$ elliptic fibration  \cite{EFY}.

\subsection{Whitney branes and Orientifold Euler characteristic}
Since Whitney branes are singular, we have to be careful how we define  their Euler characteristic or more generally their total Chern class. The appropriate definition has been worked out in \cite{CDE,AE1}. For a Whitney brane $D_w$, we have 
\begin{equation}\label{Withney.Chern}
c(D_w):= \rho_* c(\overline{D}_w)-c(S),
\end{equation}
where $\overline{D}_w$ is the normalization of $D_w$ and $S$ is the locus of codimension-two singularities of the Whitney brane. The corresponding Euler characteristic is known as the {\em orientifold Euler characteristic} $\chi_o(D)=\chi(\overline{D})-\chi(S)$. 
The original weak coupling limit discussed by Sen will satisfy the F theory-type IIB  tadpole matching condition only thanks to the presence of the singularities of the Whitney brane \cite{CDE,AE1}.  The orientifold Euler characteristic is also useful for certain weak coupling limits of E$_7$ and D$_5$ elliptic fibrations \cite{AE2,EFY}.

\section{Weak coupling limit of a Q$_7(\mathscr{L},\mathscr{S})$ model \label{section.weak}}

Following \cite{AE2}, we characterize geometrically a weak coupling limit by a transition from a semi-stable to an unstable fiber. 
The transition that we will consider is between a fiber of Kodaira type I$_2$ to a fiber of Kodaira type III. A fiber of type I$_2$ is composed of a conic intersecting a line at two distinct points. It specializes to a fiber of type III when the line is tangent to the conic, that is, when the two intersection points coincide.
As reviewed in table \ref{Table.E6'.fibers}, a fiber of type I$_2$ is characterized by  $b_2=c_3=0$ and it specializes to a fiber of typer III when in addition  $c_2=0$. 
We will use $\epsilon$ as our deformation parameter and the weak coupling limit ($j\rightarrow \infty$) will be reached as $\epsilon$ approaches zero. 
We will also denote  $h$ a section of $\mathscr{L}^2$.  We will use it to define the double cover of the base $\rho: X\rightarrow B$ as $\xi^2=h$. 

\subsection{Choice of a weak coupling limit}
We will impose the fiber I$_2$ in the weak coupling limit (that is at $\epsilon=0$)  and the fiber III over the orientifold at $\epsilon=h=0$:
$$
\begin{tabular}{ccc}
I$_2$ &  & \   III\\
{\begin{tikzpicture} [scale=.4]  \KodDeux    \end{tikzpicture}} & $\overset{\longrightarrow}{\phantom{x}}$ & {\begin{tikzpicture} [scale=.4]   \KodairaIIICubic \end{tikzpicture}}\\
$\epsilon=0$ & & $\epsilon=h=0$
 \end{tabular}
$$
 This is done by the following choice:
\begin{equation}\label{WCL.I2.III}
\text{Weak coupling limit: I$_2\rightarrow$ III}\quad \begin{cases}  b_2 = \epsilon^2 \rho , \quad  \quad c_3 =\epsilon k,\\
 c_0 =\chi ,\quad \quad\   \       c_1 =2\eta,\quad  c_2 = h,\end{cases}
\end{equation}
which leads to  the following behavior  at leading order in $\epsilon$
\begin{equation}
\Delta\propto  \epsilon^2  h^2 k^2 (\eta^2-h \chi), \quad j \propto \frac{h^4}{\epsilon^2 k^2 (\eta^2-h\chi)}.
\end{equation}
It is then direct to see that we do have a weak coupling limit since  $\lim_{\epsilon \rightarrow 0} j =\infty$ as long as we are away from $h=0$. 
Over a general point of $h$, we have  $\lim_{\epsilon \rightarrow_ 0} j=0$. But at the intersection of $h=0$ with $k(\eta^2-\chi h)=0$, the $j$-invariant is not well-defined. 

\subsection{Brane spectrum at weak coupling}
Defining the double cover $\rho: X\rightarrow B$ branched over the locus $h=0$ :
\begin{equation}
X: \  \xi^2=h.
\end{equation}
With the weak coupling limit \eqref{WCL.I2.III},  we can identify the following spectrum at  weak coupling:
\begin{enumerate}
\item $O: \xi=0$:  the orientifold,
\item $D:  k=0$: a  stack of two invariant D7-branes,
\item $D_w:  \eta^2-\xi^2 \chi=0$,  a Whitney-brane.  
\end{enumerate}

\subsection{Topological tadpole matching for Q$_7(\mathscr{L},\mathscr{S})$ models}
The weak coupling limit  \eqref{WCL.I2.III} constructed in the previous subsection naturally leads to the following relation 
$$
2\chi(Y)=4\chi(O)+2\chi(D)+\chi(D_w),
$$
which is a direct consequence of the following theorem:
\begin{theorem}[Topological tadpole matching for Q$_7(\mathscr{L},\mathscr{S})$ elliptic fibrations]
A Q$_7(\mathscr{L},\mathscr{S})$ elliptic fibration endowed with the weak coupling  limit \eqref{WCL.I2.III} satisfies the topological tadpole matching condition at the level of the total Chern class: 
$$
2\varphi_* c(Y)=4\rho_* c(O)+2\rho_* c(D)+\rho_*c^\infty(D_w),
$$
where the Chern class of the Whitney brane is understood as $\rho_*c^\infty(D_w)=\rho_*c(\overline{D}_w)-\rho_*c(S)$, with $\overline{D}_w$ the normalization of $D_w$ and $S$ the cuspidial locus of the Whitney brane.
\end{theorem}

First we establish the following important lemma that gives the Chern class from which we compute the orientifold Euler characteristic of a Whitney brane.

\begin{lemma}\label{LemmaWhitney}
Consider $\underline{D}:\eta^2- h \chi=0$ 
\begin{align}
\rho_* c^\infty(D_w) &=
\frac{4( 2 L - S)}{(1 + 2 L) (1 + 2 L - 2 S)}.
\label{WCC}
\end{align}
\end{lemma}
\begin{proof}This is a direct calculation following \cite{AE1}:
\begin{align}
\rho_* c^\infty(D_w) &=2c_{SM}(\underline{D})-2i^* (S) \quad \text{(by definition)} \nonumber\\
&= \frac{4( 2 L - S)}{(1 + 2 L) (1 + 2 L - 2 S)}.
\end{align}
\end{proof}

We can now prove the theorem.

\begin{proof}
The Chern class of the Whitney brane is understood as $\rho_*c^\infty(D_w)=\rho_*c(\overline{D}_w)-\rho_*c(S)$, with $\overline{D}_w$ the normalization of $D_w$ and $S$ the cuspidial locus of the Whitney brane.
\begin{align}
\pi_\ast c(Y) &=
\frac{6 (2 L + 2 L^2 - L S +S^2)}{(1 + 2 L - 2 S) (1 + 2 L + S)} c(B),\\
\rho_\ast c(O) &=\frac{2L}{1+2L}c(B),\\
\rho_\ast c(D) &=\frac{2 (1+L)(2L+S)}{(1+2L)(1+2L+S)} c(B),\\
c^\infty(D_w) &=\frac{4( 2 L - S)}{(1 + 2 L) (1 + 2 L - 2 S)},
\end{align}
The generalized tadpole follows immediately form the following  rational identity:
\begin{equation}
\frac{12(2 L + 2 L^2 - L S - S^2)}{(1 + 2 L - 2 S) (1 + 2 L + S)} - 
  \frac{8L(1 + 2 L) - 4 (1 + L) (2 L + S)}{(1 + 2 L) (1 + 2 L + S)} - 
  \frac{4 (2 L - S)}{(1 + 2 L) (1 + 2 L - 2 S)}=0.
\end{equation}
\end{proof}

\subsection{Weak coupling geometry: a second look.}

The weak  coupling limit we have obtained previously for the elliptic fibration of type Q$_7(\mathscr{L},\mathscr{S})$ is given by the following  family over the $\epsilon$-line:
\begin{equation}
Y_{(\epsilon)}: y ( x^2 + \chi y^2 + 2\eta y z+   h   z^2)+ \epsilon(  k  z^3 +\epsilon \rho x z^2)=0.
\end{equation}
When  $\epsilon\neq 0$, $Y_{(\epsilon)}$ is a smooth elliptic fibration. When  $\epsilon=0$, $Y_{(\epsilon)}$  degenerates into the normal crossing variety 
\begin{equation}
Y_{(0)}:\quad  y ( x^2 + \chi y^2 + 2\eta y z+   h   z^2)=0,
\end{equation}
which is composed of two smooth varieties $Z_1$ and $Z_2$ :
\begin{equation}
Z_1: y=0, \quad\quad Z_2:  x^2 + \chi y^2 + 2\eta y z+   h   z^2=0.
\end{equation}
$Z_1$ is the bundle $\mathbb{P}^1[\mathscr{L}\oplus \mathscr{O}_B]$  over the base $B$  while $Z_2$ is a fibration of conics realized as quadric in the  $\mathbb{P}^2$-bundle in which the elliptic fibration is defined. 
The normal crossing variety  $Y_{(0)}$ is  a fibration of intersecting  $\mathbb{P}^1$s whose generic fiber is a fiber of Kodaira type  I$_2$ realized by a line ( the fiber of $Z_1$)  intersecting transversally  a conic ( the fiber of $Z_2$). 
\begin{lemma}
The intersection of the two irreducible components of $Y_{(0)}$ is a smooth  variety  $X$ which is a double cover  $\rho:X\rightarrow B$ of the base $B$.
\end{lemma}
Indeed, the intersection is defined by the following complete intersection:
\begin{equation}
X=Z_1\cap Z_2 :  y =  x^2 +   h   z^2=0.
\end{equation}
This intersection is completely included in  the patch $z\neq 0$ as otherwise $x=y=z=0$, which is not allowed  since $[x:y:z]$ are projective coordinates of a $\mathbb{P}^2$ projective bundle. In order to connect with notations familiar in F-theory, we  put $y=z-1=0$ in the definition of $Z_1\cap Z_2$. We are left only with the affine coordinate $x$, which is a section of the line bundle $\mathscr{L}$. Introducing  $\xi=i x$, the intersection can simply be expressed in the total space of the line bundle $\mathscr{L}$ by the canonical equation $\xi^2=h$:
\begin{equation}
\rho:  X\rightarrow B:\quad  \xi^2=h.
\end{equation}
This is a   double cover of the base $B$ of the original elliptic fibration branched on the divisor $\underline{O}: h=0$ in $B$: 
\begin{equation}
\underline{O}: h=0.
\end{equation}
This divisor of $B$ pulls back to the divisor $O=\rho^\star \underline{O}$:
\begin{equation}
O:\xi=0, \quad \text{in}\quad  X.
\end{equation}
 It is the divisor $O\subset X$ which is called the {\em orientifold plane}.  
One can see $X$ as a type IIB orientifold weak coupling limit of F-theory on $Y$. 
The   discriminant locus of the conic fibration $Z_2$:
$$
\underline{D}_w : \eta^2- h \chi =0,
$$
which is singular at $\eta=h=\chi=0$. The variety $\underline{D}_w$  pull-back in the double cover $X$ as the Whitney brane $D_w=\rho^\star \underline{D}_w$: 
\begin{equation}
D_w: \quad \eta^2 - \xi^2 \chi=0.
\end{equation}
\subsection{Fiberwise description of the limit}
The generic fiber at $\epsilon\neq 0$ is a smooth elliptic curve. At $\epsilon=0$, the fiber degenerates to a singular elliptic fiber of Kodaira  type  I$_2$. This I$_2$ fiber is composed of a line and a conic meeting at two distinct points. 
The leading order in  $\epsilon$ defines a family of elliptic fibration  
$$Y_{(\epsilon)}:  x^2 y+ \chi y^3 + 2\eta y^2 z+   h  y z^2+ \epsilon  k  z^3 =0.$$
At leading order in $\epsilon$, the  discriminant locus of $Y_{(0)}$  splits into three components 
\begin{equation}
\Delta_{(\epsilon)}\propto  \quad  \epsilon ^2 h^2  k^2 \left(  \eta ^2- h \chi \right)=0.
\end{equation}
The first one is the branch locus of the orientifold $h=0$, the second one $\underline{D}: k=0 $ is a stack of two branes transversal to the orientifold and the last one is a Whitney brane $\underline{D}: \eta^2-\chi h=0$. 
The fibers over $h=0$ are of type I$_1$ for $\epsilon \neq 0$ and of type III for $\epsilon=0$. Over $k=0$, the fibers are of type I$_2$. 
The fiber over the Whitney brane are of type I$_1$ when $\epsilon \neq 0$ and of type I$_3$ when $\epsilon=0$. 
If we consider higher terms in $\epsilon$, the stack of branes and the Whitney brane recombine into a unique brane:
 \begin{equation}
\Delta_{(\epsilon)}\propto \epsilon ^2 k^2 \left(-4 h^2 \eta ^2+4 h^3 \chi+32 k \epsilon  \eta ^3 -36 h k \epsilon  \eta  \chi +27 k^2 \epsilon ^2 \chi ^2\right)=0.
\end{equation}
Interestingly,  $Y_{(0)}$ is not an elliptic fibration: the generic fiber is not an elliptic curve but a fiber of type  I$_2$, composed of a line $y=0$ and a conic  $x^2+\chi y^2 + 2 \eta y z + h z^2=0$. The fiber I$_2$ specializes to a fiber of type III when the line becomes tangent to the conic. This happens when $h=0$:  
\begin{align}
\underline{O}:  & \quad h=0\rightarrow \text{III}.\\
\intertext{The fiber I$_2$ specializes to a triangle I$_3$ as the conic splits into two lines. This happens when the discriminant of the conic vanishes and  corresponds to the Whitney brane: }
\underline{D_w}: & \quad   \eta^2-\chi h=0\rightarrow \text{I}_3.\\
\intertext{
The two lines coming from the conics are no individually well defined because of a  $\mathbb{Z}_2$ monodromy\footnote{ 
This is exactly what the double cover $X\rightarrow B$ defined by $\xi^2=h$ will do. }. When $h=\eta=0$, the fiber specializes further to a star (a fiber of type IV): }
\underline{O} \cap \underline{D_w}:& \quad    h=\eta=0\rightarrow \text{IV}.\\
\intertext{Finally, when $h=\eta=\chi=0$, the fiber specializes to a double line $x^2=0$ intersecting transversally the line $y=0$. This is a non-Kodaira fiber of type IV$^{(2)}$: }
Sing(\underline{D_w}): & \quad   h=\eta=\chi=0\rightarrow \text{IV}^{(2)}.
\end{align}

\section{Conclusion and discussion  \label{section.conclusions}}

In this paper we  introduce a new model for elliptic fibrations with a Mordell-Weil group of rank one using an hypersurface in a projective bundle. Global aspects of this elliptic fibration are controlled by two line bundles $\mathscr{L}$ and $\mathscr{S}$ that are used to define the ambient space $\pi:\mathbb{P}[\mathscr{L}\oplus\mathscr{S}\oplus\mathscr{O}_B]\rightarrow B$. The equation is then retrieved  as a section of the line bundle $\mathscr{O}(3)\otimes \pi^\ast \mathscr{L}^2\otimes \pi^\ast \mathscr{S}$. 
 The resulting defining equation is a cubic with a Newton's polygon which is  a reflexive polytope in  a  quadrilateral shape with seven lattice points on its boundary. We call it an elliptic fibration of type Q$_7(\mathscr{L},\mathscr{S})$. 
This models generalize both  the E$_6$ elliptic fibration and the elliptic fibration introduced recently by  Cacciatori, Cattaneo, and Van Geemen  \cite{CCVG}.   Using this smooth model we can easily determine the spectrum of singular fibers and compute basic topological invariants.  We identify seven possible singular fibers: six Kodaira fibers  (type I$_1$, I$_2$, I$_3$, II, III and  IV)  and the non-Kodaira fiber of  type IV$^{(2)}$. We also get a pushforward formula for the total Chern class. This is a  generalized Sethi-Vafa-Witten formula. 
Using the geometric description of the weak coupling limit developed in \cite{AE2}, we find a weak coupling limit for the Q$_7(\mathscr{L},\mathscr{S})$ elliptic fibration. The weak coupling that we consider is based on  a specialization to a fiber of type I$_2$ over a general point of the base. The fiber specialize further to a fiber of type  III over the orientifold $h=0$:
$$
\text{weak coupling limit:}\quad\quad\quad
\begin{tabular}{ccc}
I$_2$ &  & \   III\\
{\begin{tikzpicture} [scale=.4]  \KodDeux    \end{tikzpicture}} & $\overset{\longrightarrow}{\phantom{x}}$ & {\begin{tikzpicture} [scale=.4]   \KodairaIIICubic \end{tikzpicture}}\\
$\epsilon=0$ & & $\epsilon=h=0$
 \end{tabular}
$$
It yields the  following brane spectrum at weak coupling:
$$(\star)\quad
\text{a $\mathbb{Z}_2$  orientifold}+
\text{ a Whitney brane}+
\text{an $Sp(1)$ stack}.
$$
The $Sp(1)$ stack is composed of two smooth and invariant branes intersecting the orientifold transversally. 
We prove that the tadpole matching condition is satisfied for this spectrum. 
The singularities of the Whitney brane play an essential  role as they contribute to the D3-charge at weak coupling. Such a contribution is already necessary for the usual Sen's limit and for certain weak coupling limits of E$_7$ elliptic fibrations \cite{CDE,AE1,AE2}. The Euler characteristic for the Whitney brane is defined using the {\em orientifold Euler characteristic} introduced in \cite{CDE} and mathematically defined in \cite{AE1}. 
The weak coupling limit that we have obtained is naturally an ALE degeneration of the elliptic fibration. As the coupling becomes weak, the generic fiber is no longer elliptic but becomes 
a fiber of type I$_2$ composed of a conic and a line meeting at two distinct points.  
As we move over the locus of the orientifold, the conic becomes tangent to the line and we get a fiber of type III.

There are several interesting questions that are not discussed in this paper. For example, the weak coupling limit discussed here is not unique. But it is not clear that another one would satisfy the tadpole condition. 
The fibration discussed in \cite{CCVG} does not satisfy the tadpole condition but seems to admit specializations describing orientifolds with surprising properties. 
An analysis of these particular cases will be the subject of a companion paper.

\vskip .5 cm

\noindent{\bf Acknowledgments.}  It is a pleasure to thank  Paolo Aluffi, Ron Donagi, Patrick \mbox{Jefferson}, Sungkyung Kang, Denis Klevers, Daniel Park, and Shu-Heng~Shao for discussions. 
M.~J.~Kang would like to thank Daniel Jafferis for his guidance and constant support. 
M.~Esole is supported  in part  by  the National Science Foundation (NSF) grant DMS-1406925 ``Elliptic Fibrations and String Theory''. 
The research of M.~J.~Kang  was supported by the Fundamental Laws Initiative of the Center for the Fundamental Laws of Nature, Harvard University.

\appendix 

\section{Alternative derivation of the Jacobian \label{Appendix.A}} 
The Jacobian of a genus one curve with a Mordell-Weil group of rank one can be easily obtained using the Riemann-Roch theorem as discussed for example in \cite{MorrisonPark,Braun:2014oya}.
In this section, we present a ``quick and dirty trick" that reproduces the same result in an intuitive way. It also gives a simple arithmetic meaning to the section $\mathscr{S}$ that enters in the definition of the Q$_7(\mathscr{L},\mathscr{S})$.

Consider a Weierstrass model with a rational point $P$ other than the  point at infinity $O:x=z=0$. Putting $P$  at $y=0$, the cubic  $x^3+a_2 x^2+ a_4 x+a_6$ has to factorize.
By an appropriate translation of $x$, we can put the Weierstrass model in the following form:
\begin{equation}
y(y+a_1 x + a_3)=(x+a_2)(x^2+a_4).
\end{equation}
We end up with a general Weierstrass model with 
 the specialization $a_6=a_2a_4$: 
\begin{equation}
y^2 + a_1 x y + a_3 y=x^3 + a_2 x^2 + a_4  x + a_2 a_4.
\end{equation}
This specialization is too mild to factorize the discriminant but does gives two non-trivial rational sections:  
\begin{equation}
(x,y)=(-a_2, 0) \  \text{ and }\   (x,y)=(-a_2, a_1 a_2-a_3).
\end{equation}
 These two points are inverse of each other for the Mordell-Weil group with the neutral element $x=z=0$.  
 Since we can write the Weierstrass model as: 
\begin{equation}
y(y+ a_1 x + a_3 ) =(x+a_2)(x^2+a_4),
\end{equation}
we have conifold-like  points at $y=y+a_1 x + a_3=x+a_2=x^2+a_4=0$.
It corresponds to the point  $y=x+a_2=0$ on each fiber over  the codimension-2 locus in the base:
\begin{equation}
a_3-a_1 a_2=a_4+a_2^2=0.
\end{equation}
Over this locus, the elliptic fiber can be put in this suggestive form: 
\begin{equation}
a_3-a_1 a_2=a_4+a_2^2=0\Longrightarrow (y+\tfrac{a_1}{2}( x + a_2))^2=(x+a_2)^2(x-a_2+\frac{a_1^2}{4}),
\end{equation} 
which shows that the elliptic curve has a $A_1$ singularity over  $a_3-a_1 a_2=a_4+a_2^2=0$. 

We now consider the case in which  $a_2$ has an explicit rational part:
\begin{equation}
a_2=c-\frac{p}{q},
\end{equation}
where $c$ is integral and $p/q$ is a reduced fraction. 
By taking appropriate choices for $(a_1,a_3,a_4)$, we can still ensure that when we complete the square in $y$, the coefficients $b_2$, $b_4$ and $b_6$ are all integral. 
Since $b_2=4a_2+a_1^2$, we can get rid of the fractional part of $b_2$ due to $a_2$ by taking $a_1=2\sqrt{p/q}$. 
However, the point $(x,y)=(a_2,-a_1 a_2+ a_3)$ will not be rational anymore  because of the square root in $a_1$. We can solve this problem by requiring that $p/q$ to be  a perfect square ($p/q=r^2/s^2$). That is: 
\begin{equation}
a_1=2 \frac{r}{s}, \quad a_2=c-\frac{r^2}{s^2}. 
\end{equation}
We can get rid of the fractional part of $b_4=2a_4+a_1 a_3$  by requiring $a_3$ to be proportional to $s$:
\begin{equation}
a_1=2 \frac{r}{s}, \quad a_2=c-\frac{r^2}{s^2}, \quad a_3= 2 s t.
\end{equation}
Since  $b_6=4 a_6 +a_3^2$ and $a_6=a_2 a_4$,  we can ensure that  $b_6$ is integral by taking  $a_4$ to be proportional to $s^2$. Using $a_6=a_2 a_4$, we get our final form: 
\begin{equation}
a_1=2 \frac{r}{s}, \quad a_2=c-\frac{r^2}{s^2}, \quad a_3=2 s t, \quad a_4=  
s^2 u, \quad a_6=u (c s^2-r^2).
\end{equation}
\begin{equation}
E:\quad  y^2 + 2\frac{r}{s} x y +2 st y= x^3+(c-\frac{r^2}{s^2}) x^2 + 
s^2 u x + u(cs^2 -r^2).
\end{equation}
This has a rational point of type $(x,y)=(-a_2, 0)$ with $a_2=c-r^2/s^2$.
This Weierstrass equation has  coefficients that are rational expressions. But by construction, we can resolve  this problem by completing the square in $y$:
\begin{equation}
E:\quad y(y  +2st)= x^3+c x^2 + 
(s^2 u +2 rt )x + u(cs^2 -r^2).
\end{equation}
with the rational point
\begin{equation}
 x=-c+ \frac{r^2}{s^2}, \quad y= \frac{r^3-rc s^2 + s^4 t }{s^3}. 
 \end{equation}
Completing the square in $y$ we get the canonical form of a Weierstrass model of rank 1: 
\begin{equation}
E:\quad y^2 = x^3+c x^2 + 
(s^2 u+ 2 r t) x + ucs^2 - ur^2+ s^2 t^2,
\end{equation}
The corresponding  short Weierstrass form is then 
\begin{equation}\label{JacR1}
y^2=x^3+(-\tfrac{1}{3}c^2+2 r t+s^2 u)x+(\tfrac{2 }{27}c^3+s^2 t^2-r^2 u-\frac{2}{3} c (r t-s^2 u)),
\end{equation}
which depends only on the variable $(s^2, r,c,t,u)$. The reduced Weierstrass model is invariant under the involution:
\begin{equation}
(s^2, r,c,t,u)\leftrightarrow  (-u, t,c,r,-s^2).
\end{equation}

It is also invariant under the scaling symmetry
\begin{equation}
\alpha\cdot(s^2, r,c,t,u)=  (\alpha^2 s^2,\alpha  r,c,\frac{t}{\alpha},\frac{u}{\alpha^2}).
\end{equation}

It is the Jacobian of the Jacobi quartic:
\begin{equation}\label{EqQ1}
y^2 =s^2 x^4 - 2 r x^3  z+cx^2 z^2+  tx z^3 +\tfrac{1}{4}  u z^4.
\end{equation}
or equivalently
\begin{equation}\label{EqQ2}
y^2 =-u x^4 - 2 t x^3 z +cx^2 z^2+  rx z^3 -\tfrac{1}{4}  s^2 z^4.
\end{equation}
Both Jacobi quartics  \eqref{EqQ1} and \eqref{EqQ2} admit the same Jacobian.

In a Weierstrass model, each coefficient  $a_i$ ($i=1,2,3,4,6$) is a section of  $\mathscr{L}^i$, where $\mathscr{L}$ is the fundamental line bundle of the elliptic fibration.   Assuming that  $s$ is a section of a line bundle $\mathscr{M}$, then the different variables $r,c,t,u$ are sections of the following line bundles:

\vspace{.2cm}
\begin{center}
\begin{tabular}{|l|c|c|c|c|c|}
\hline 
\rule{0pt}{3ex}
  Line bundle &   $\mathscr{M}^2$ & $\mathscr{L}\otimes \mathscr{M}$ &  $\mathscr{L}^2$ &  $\mathscr{L}^3\otimes \mathscr{M}^{-1}$ & $\mathscr{L}^4\otimes \mathscr{M}^{-2}$\\
\hline 
 Section & $s^2$ & $r$ & $c$ & $t$ & $u$\\
\hline
\end{tabular}
\end{center}
{\vskip 2mm}

In this table,   as we move to the right, we multiply the line bundle by $\mathscr{L}\otimes\mathscr{M}^{-1}$. 
The dictionary to the notation in the main text is:
\begin{equation}
s=b_2,  \quad r= c_3,\quad   c=-c_2, \quad  \quad t=-\frac{1}{2} c_1, \quad u=-c_0.
\end{equation}
We see that the line bundle $\mathscr{S}$ corresponds to $\mathscr{L}\otimes\mathscr{M}$. This shows that 
$\mathscr{S}\otimes \mathscr{L}^{-1}$ is a natural line to consider to discuss the arithmetic properties of the section.


\thebibliography{99}
\bibitem{Ball}
W.~W.~.Ball, {\em Newton's classification of cubic curves}. Proc. London Math. Soc. 22  (1890), 104-143.
\bibitem{Talbot}  C.~.M.~.R.~Talbot, {\em Sir Isaac Newton's Enumeration of Lines of the Third Order, Generation of Curves by
Shadows, Organic Description of Curves, and Construction of Equations by Curves}. London. Translation
and commentary by C. R. M. Talbot (1860).
\bibitem{Guicciadini}
N.~Guicciadini, {\em Isaac Newton On mathematical Certainty and Method}, Cambridge, MA, MIT Press, 2009.
\bibitem{Lexicom}
J.~Harris, {\em Lexicom Technicum}, London, 1710. 
\bibitem{CCVG}
S.~L.~Cacciatori, A.~Cattaneo and B.~Geemen,
  ``A new CY elliptic fibration and tadpole cancellation,''
  JHEP {\bf 1110}, 031 (2011)
  [arXiv:1107.3589 [hep-th]].  
  \bibitem{MorrisonPark} 
  D.~R.~Morrison and D.~S.~Park,
  ``F-Theory and the Mordell-Weil Group of Elliptically-Fibered Calabi-Yau Threefolds,''
  JHEP {\bf 1210}, 128 (2012)
  [arXiv:1208.2695 [hep-th]].

\bibitem{MumS}
D.~Mumford and K.~Suominen, {\em Introduction to the theory of moduli}. In Algberaic Geometry, Oslo 1970 (Proc. Fith Nordic Summer-School in Math.), 
Wolters-Noordhoff, Groningen, 1972, 171-222.

\bibitem{Formulaire}
P.~Deligne, {\it Courbes elliptiques: formulaire d'apr{\`e}s J. Tate},  in
  {\em Modular functions of one variable}, IV (Proc. Internat. Summer
  School, Univ. Antwerp, Antwerp, 1972), pp.~53--73. Lecture Notes in
  Math., Vol. 476.
\newblock Springer, Berlin, 1975.
\bibitem{Tate} J.T.~ Tate, {\it The Arithmetics of Elliptic Curves}, Inventiones math. 23, 170-206 (1974).
\bibitem{Hermite} Ch.~Hermite, ``Sur la th\'eorie des fonctions homog\'enes \`a deux ind\'etermin\'ees", Premier 
M\'emoire, J.~Reine Angew.\ Math.\ 52 (1856), pp. 1-17.
\bibitem{Weil.Hermite}
A.~Weil, ``Remarques sur un m\'emoire d'Hermite", Arch. Math. 5 (1954), pp 197-202.
\bibitem{Cassels} J.~W.~Cassels, ``Lectures on Elliptic Curves", London Mathematical Society Student Texts 24, Cambridge University Press, Cambridge, 1991.

\bibitem{Weil.History} A.~Weil, ``Number Theory, An Approach through History, from Hammurapi to Legendre,``  Boston. MA: Birkh\"auser Boston Inc. , 1984.

\bibitem{ARVT}
M.~Artin, F.~Rodriguez-Villegas, J.~Tate, ``On the Jacobians of plane cubics", Advances in Mathematics, Volume 198, Issue 1, 1 December 2005, Pages 366-382.
\bibitem{BLR} S.~Bosch, W.~L\"{u}tkebohmert,  M.~Raynaud, N\'eron Models, Springer, Berlin, 1990.
\bibitem{Sen.Orientifold}
  A.~Sen,
  ``Orientifold limit of F theory vacua,''
  Phys.\ Rev.\  {\bf D55}, 7345-7349 (1997).
  [hep-th/9702165].
\bibitem{CDE}
  A.~Collinucci, F.~Denef, M.~Esole,
  ``D-brane Deconstructions in IIB Orientifolds,''
  JHEP {\bf 0902}, 005 (2009).
  [arXiv:0805.1573 [hep-th]].
  
\bibitem{Denef.LH} 
  F.~Denef,
  ``Les Houches Lectures on Constructing String Vacua,''
  arXiv:0803.1194 [hep-th].

\bibitem{Esole:2012tf} 
  M.~Esole and R.~Savelli,
  ``Tate Form and Weak Coupling Limits in F-theory,''
  JHEP {\bf 1306}, 027 (2013)
  [arXiv:1209.1633 [hep-th]].

  
\bibitem{AE1}  P.~Aluffi, M.~Esole,
  ``Chern class identities from tadpole matching in type IIB and F-theory,''JHEP {\bf 0903}, 032 (2009).
  [arXiv:0710.2544 [hep-th]].
\bibitem{AE2}
  P.~Aluffi, M.~Esole,  ``New Orientifold Weak Coupling Limits in F-theory,'' JHEP {\bf 1002}, 020 (2010).
  [arXiv:0908.1572 [hep-th]].
\bibitem{EFY}
  M.~Esole, J.~Fullwood, S.~-T.~Yau,
  ``D5 elliptic fibrations: non-Kodaira fibers and new orientifold limits of F-theory,''
    [arXiv:1110.6177 [hep-th]].
    \bibitem{EY}
  M.~Esole and S.~T.~Yau,
  ``Small resolutions of SU(5)-models in F-theory,''
  Adv.\ Theor.\ Math.\ Phys.\  {\bf 17}, 1195 (2013)
  [arXiv:1107.0733 [hep-th]].    \bibitem{Esole:2014hya} 
  M.~Esole, S.~H.~Shao and S.~T.~Yau,
  ``Singularities and Gauge Theory Phases II,''
  arXiv:1407.1867 [hep-th].
  \bibitem{Esole:2014bka} 
  M.~Esole, S.~H.~Shao and S.~T.~Yau,
  ``Singularities and Gauge Theory Phases,''
  arXiv:1402.6331 [hep-th].
  \bibitem{Morrison:2011mb} 
  D.~R.~Morrison and W.~Taylor,
  ``Matter and singularities,''
  JHEP {\bf 1201}, 022 (2012)
  [arXiv:1106.3563 [hep-th]].
  
  \bibitem{Hayashi:2014kca} 
  H.~Hayashi, C.~Lawrie, D.~R.~Morrison and S.~Schafer-Nameki,
  ``Box Graphs and Singular Fibers,''
  JHEP {\bf 1405}, 048 (2014)
  [arXiv:1402.2653 [hep-th]].
  
  \bibitem{Cattaneo:2013vda} 
  A.~Cattaneo,
  ``Elliptic fibrations and the singularities of their Weierstrass models,''
  arXiv:1307.7997 [math.AG].
  
  \bibitem{Braun:2013cb} 
  A.~P.~Braun and T.~Watari,
  ``On Singular Fibres in F-Theory,''
  JHEP {\bf 1307}, 031 (2013)
  [arXiv:1301.5814 [hep-th]].
%
\bibitem{Vafa:1996xn}
  C.~Vafa,
  ``Evidence for F theory,''
  Nucl.\ Phys.\  {\bf B469}, 403-418 (1996).
  [hep-th/9602022].
  \bibitem{Morrison:1996na}
  D.~R.~Morrison, C.~Vafa,
  ``Compactifications of F theory on Calabi-Yau threefolds. 1,''
  Nucl.\ Phys.\  {\bf B473}, 74-92 (1996).
  [hep-th/9602114].
\bibitem{GVW} 
  S.~Gukov, C.~Vafa and E.~Witten,
  ``CFT's from Calabi-Yau four folds,''
  Nucl.\ Phys.\ B\ {\bf 584}, 69  (2000)
  [Erratum-ibid.\ B\ {\bf 608}, 477  (2001)]
  [hep-th/9906070].

\bibitem{Weil.NTH} A.~Weil, ``Number Theory:  An approach through history:  From Hammurapi to Legendre," Boston,  Birkh\"auser, 1984.

\bibitem{Poincare.Rank} H.~Poincar\'e , ``Sur les propri\'et\'es arithm\'etiques des courbes alg\'ebriques", Journal de 
Liouville ~7, 1901, p.~161.

  
\bibitem{Park:2011wv} 
  D.~S.~Park and W.~Taylor,
  ``Constraints on 6D Supergravity Theories with Abelian Gauge Symmetry,''
  JHEP {\bf 1201}, 141 (2012)
  [arXiv:1110.5916 [hep-th]].


\bibitem{Park:2011ji} 
  D.~S.~Park,
  ``Anomaly Equations and Intersection Theory,''
  JHEP {\bf 1201}, 093 (2012)
  [arXiv:1111.2351 [hep-th]].


\bibitem{Cvetic:2013qsa} 
  M.~Cvetic, D.~Klevers, H.~Piragua and P.~Song,
  ``Elliptic fibrations with rank three Mordell-Weil group: F-theory with U(1) x U(1) x U(1) gauge symmetry,''
  JHEP {\bf 1403}, 021 (2014)
  [arXiv:1310.0463 [hep-th], arXiv:1310.0463].


\bibitem{Braun:2014oya} 
  V.~Braun and D.~R.~Morrison,
  ``F-theory on Genus-One Fibrations,''
  JHEP {\bf 1408}, 132 (2014)
  [arXiv:1401.7844 [hep-th]].


\bibitem{Morrison:2014era} 
  D.~R.~Morrison and W.~Taylor,
  ``Sections, multisections, and U(1) fields in F-theory,''
  arXiv:1404.1527 [hep-th].


\bibitem{Mayrhofer:2014opa} 
  C.~Mayrhofer, D.~R.~Morrison, O.~Till and T.~Weigand,
  ``Mordell-Weil Torsion and the Global Structure of Gauge Groups in F-theory,''
  arXiv:1405.3656 [hep-th].


\bibitem{Kuntzler:2014ila} 
  M.~Kuntzler and S.~Schafer-Nameki,
  ``Tate Trees for Elliptic Fibrations with Rank one Mordell-Weil group,''
  arXiv:1406.5174 [hep-th].


\bibitem{Klevers:2014bqa} 
  D.~Klevers, D.~K.~Mayorga Pena, P.~K.~Oehlmann, H.~Piragua and J.~Reuter,
  ``F-Theory on all Toric Hypersurface Fibrations and its Higgs Branches,''
  arXiv:1408.4808 [hep-th].


\bibitem{Braun:2014nva} 
  A.~P.~Braun, A.~Collinucci and R.~Valandro,
  ``The fate of U(1)'s at strong coupling in F-theory,''
  JHEP {\bf 1407}, 028 (2014)
  [arXiv:1402.4054 [hep-th]].

\bibitem{Cvetic:2013jta} 
  M.~Cvetic, D.~Klevers and H.~Piragua,
  ``F-Theory Compactifications with Multiple U(1)-Factors: Addendum,''
  JHEP {\bf 1312}, 056 (2013)
  [arXiv:1307.6425 [hep-th]].


\bibitem{Borchmann:2013hta} 
  J.~Borchmann, C.~Mayrhofer, E.~Palti and T.~Weigand,
  ``SU(5) Tops with Multiple U(1)s in F-theory,''
  Nucl.\ Phys.\ B {\bf 882}, 1 (2014)
  [arXiv:1307.2902 [hep-th]].


\bibitem{Cvetic:2013uta} 
  M.~Cvetic, A.~Grassi, D.~Klevers and H.~Piragua,
  ``Chiral Four-Dimensional F-Theory Compactifications With SU(5) and Multiple U(1)-Factors,''
  JHEP {\bf 1404}, 010 (2014)
  [arXiv:1306.3987 [hep-th]].


\bibitem{Braun:2013nqa} 
  V.~Braun, T.~W.~Grimm and J.~Keitel,
  ``Geometric Engineering in Toric F-Theory and GUTs with U(1) Gauge Factors,''
  JHEP {\bf 1312}, 069 (2013)
  [arXiv:1306.0577 [hep-th]].


\bibitem{Grimm:2013oga} 
  T.~W.~Grimm, A.~Kapfer and J.~Keitel,
  ``Effective action of 6D F-Theory with U(1) factors: Rational sections make Chern-Simons terms jump,''
  JHEP {\bf 1307}, 115 (2013)
  [arXiv:1305.1929 [hep-th]].


\bibitem{Cvetic:2013nia} 
  M.~Cvetic, D.~Klevers and H.~Piragua,
  ``F-Theory Compactifications with Multiple U(1)-Factors: Constructing Elliptic Fibrations with Rational Sections,''
  JHEP {\bf 1306}, 067 (2013)
  [arXiv:1303.6970 [hep-th]].

\bibitem{Borchmann:2013jwa} 
  J.~Borchmann, C.~Mayrhofer, E.~Palti and T.~Weigand,
  ``Elliptic fibrations for $SU(5)\times U(1)\times U(1)$ F-theory vacua,''
  Phys.\ Rev.\ D {\bf 88}, no. 4, 046005 (2013)
  [arXiv:1303.5054 [hep-th]].


\bibitem{Braun:2013yti} 
  V.~Braun, T.~W.~Grimm and J.~Keitel,
  ``New Global F-theory GUTs with U(1) symmetries,''
  JHEP {\bf 1309}, 154 (2013)
  [arXiv:1302.1854 [hep-th]].


\bibitem{Mayrhofer:2012zy} 
  C.~Mayrhofer, E.~Palti and T.~Weigand,
  ``U(1) symmetries in F-theory GUTs with multiple sections,''
  JHEP {\bf 1303}, 098 (2013)
  [arXiv:1211.6742 [hep-th]].
 

\bibitem{Cvetic:2012xn} 
  M.~Cvetic, T.~W.~Grimm and D.~Klevers,
  ``Anomaly Cancellation And Abelian Gauge Symmetries In F-theory,''
  JHEP {\bf 1302}, 101 (2013)
  [arXiv:1210.6034 [hep-th]].
   
  \bibitem{Braun:2011zm} 
  A.~P.~Braun, A.~Collinucci and R.~Valandro,
  ``G-flux in F-theory and algebraic cycles,''
  Nucl.\ Phys.\ B {\bf 856}, 129 (2012)
  [arXiv:1107.5337 [hep-th]].
  
  
  
  \bibitem{Clingher:2012rg} 
  A.~Clingher, R.~Donagi and M.~Wijnholt,
  ``The Sen Limit,''
  arXiv:1212.4505 [hep-th].
  
  \bibitem{Sethi:1996es} 
  S.~Sethi, C.~Vafa and E.~Witten,
  ``Constraints on low dimensional string compactifications,''
  Nucl.\ Phys.\ B {\bf 480}, 213 (1996)
  [hep-th/9606122].
\end{document}